\documentclass[11pt,a4paper]{article}
\usepackage[margin=2.6cm]{geometry}
\usepackage{amsmath,amssymb,amsthm}
\usepackage{booktabs}
\usepackage{array}
\usepackage{graphicx}
\usepackage{xcolor}
\usepackage[colorlinks=true,linkcolor=blue!60!black,citecolor=blue!60!black,urlcolor=blue!60!black]{hyperref}
\usepackage{microtype}
\hypersetup{
  pdftitle={The CKM sector of the exceptional-Jordan programme: finite-Dirac mass moduli, two conditional angle estimates, and the weak-to-mass bridge},
  pdfauthor={Tejinder P. Singh}
}

\newtheorem{theorem}{Theorem}
\newtheorem{lemma}{Lemma}
\newtheorem{proposition}{Proposition}
\theoremstyle{definition}
\newtheorem{postulate}{Postulate}
\theoremstyle{remark}

\newcommand{\Sym}{\mathrm{Sym}}
\newcommand{\Dtag}{[\textsf{D}]}
\newcommand{\Ptag}{[\textsf{P}]}
\newcommand{\Otag}{[\textsf{O}]}
\newcommand{\Vus}{|V_{us}|}
\newcommand{\Vcb}{|V_{cb}|}
\newcommand{\Vub}{|V_{ub}|}
\newcommand{\Vtd}{|V_{td}|}
\newcommand{\Vts}{|V_{ts}|}

\title{\bf The CKM sector of the exceptional-Jordan programme:\\ finite-Dirac mass moduli, two conditional angle estimates,\\ and the weak-to-mass bridge}
\author{Tejinder P. Singh\\[4pt]
\normalsize Tata Institute of Fundamental Research, Homi Bhabha Road, Mumbai 400005, India\\
\normalsize \texttt{tpsingh@tifr.res.in}}
\date{August 3, 2026}

\begin{document}
\maketitle

\begin{abstract}
\noindent
The standard model does not determine quark masses or the CKM matrix.  In the exceptional-Jordan programme, charged-fermion square-root masses occupy short $\Sym^3(\mathbf3)$ chains.  Compressing the multiplicative symmetric-cube lift onto the occupied nodes gives a root-mass operator $\mathcal R_f$ whose positive square has exactly the proposed mass-ratio spectrum.  This packages the proposed mass spectrum and the relative left frames in one finite Dirac operator; it does not derive the latter from singular-value decomposition alone.  Conditional on the occupied-node transport, virtual-node amplitude, real $(2,3)$ blocks and balanced-quadrature choices, the no-fit layer gives $\Vus=0.2371$ ($5.3\%$ high) and $\Vcb=0.0422$ ($0.8\%$ high) against the 2026 Particle Data Group global fit; the minimal identity $\Vub/\Vcb=\sqrt{m_u/m_c}$ is a factor two low, and one complex long edge is fitted to the remaining CKM angle and phase.  The scale-dependent parent mass-ratio inputs are interpreted at the common matching baseline $\mu_*=M_Z$.

For an adopted $F_1$--$F_2$ family embedding, we construct an exact Peirce-changing lift of the local Cabibbo bridge.  For an Albert frame $p_i$ and Peirce fibres $F_i(\mathbb O_{\mathbb C})$, ${\cal D}_{12}(z)=2[L_{p_1-p_2},L_{F_3(z)}]\in\mathfrak f_4(\mathbb C)$, with $z=-g_\chi$, gives the conjugate up/anti-down transport and preserves $\varphi_{12}=-2\chi$.  In the split real form, the refinement $J_2(\mathbb O_s)=J_2(\mathbb H_s)\oplus\mathbb H_s\ell$ decomposes the Peirce ten-vector as $\mathbf{10}=\mathbf6_{(3,3)}\oplus\mathbf4_{(2,2)}$.  The Albert cubic couples the four-component complement to opposite $\operatorname{Spin}(3,3)$ Weyl modules.  It therefore supplies a candidate universal chiral-soldering channel, conditional on the real-form embedding and localization.  If the finite bridge factorizes as $Y_f=B\otimes\Xi_f$ with a common nonzero $B$, this universal factor cancels from the family eigenvectors and $V_{\rm CKM}=U_{\Xi_u}^\dagger U_{\Xi_d}$.  Thus neither $B$ nor the corner scalar fixes mixing; the family kernel $\Xi_f$ remains dynamical.

The two balanced-rotor orientations give two discrete fitted long-edge branches; the commonly displayed branch is not dynamically selected.  For the adopted cyclic Majorana placement and a real-linear projected bridge, its completion has precisely the real support $(e_4,e_3,e_6)$ while the quark-phase $e_1$ quadrature drops out, making $J_\ell=0$ and $\delta^\ell_{CP}\in\{0,\pi\}$ compatible with this class.  The minimal radial-quartic cyclic truncation considered here has equal-magnitude full-rank extrema, or a flat direction when its cubic vanishes.  A mixed Albert cubic gives a stable alignment within the chosen three-edge subspace $W$ along $P_WH_q^\#$, but neither $W$ nor this charged-vacuum adjoint is yet calculated.

The lift also supplies a candidate ultraviolet skeleton for the nine-link textures of Arkani-Hamed, Figueiredo, Hall and Manzari: six diagonal mass links plus three directed Peirce links form a connected nine-edge graph with one cycle, and a sole balanced Cabibbo link would give a $\pi/2$ holonomy.  This is not yet a derived Yukawa texture: the chiral projection and right-frame lock remain open, and exponentiating the sparse generator is generally dense.  The result is an explicit Cabibbo bridge and a sharper localization of the remaining CKM/PMNS vacuum dynamics, not a completed derivation.
\end{abstract}

\section{Introduction and summary}
\label{sec:intro}

In the exceptional-Jordan programme the charged-fermion square-root masses are degree-3 monomials on the weight triangle of $\Sym^3(\mathbf 3)$ of the family $SU(3)$, arranged along minimal nearest-neighbour chains \cite{Singh2508,TeliSinghMass}. The companion mixing paper \cite{TeliSinghMixing} took the first step from spectra to mixing: an \emph{adjacent-edge lift}, in which each neighbouring pair of generations is assigned the two-state Fritzsch texture whose exact diagonalization angle is the square-root mass ratio,
\begin{equation}
\tan\theta_{ij}=\sqrt{m_i/m_j},
\label{eq:lift}
\end{equation}
together with the transport theorem of the CP Letter \cite{TeliSinghCP}: on the shared Cabibbo edge the up- and down-sector rung amplitudes are complex conjugates, $A_d=A_u^*$, so the relative Cabibbo-block phase obeys the exact local law $\varphi_{12}=-2\chi$, with the \emph{balanced} rotor $\chi=-\pi/4$ (quadrature, $\varphi_{12}=90^\circ$) as the distinguished reference point. Three quantities remained undetermined in \cite{TeliSinghMixing}: the value of $\varphi_{12}$ (extracted from the measured $\Vus$ as $105.7^\circ$), a phenomenological $(2,3)$ normalization $\kappa_{23}\simeq0.56$ (extracted from $\Vcb$), and the direct $(1,3)$ element (``a long-edge bridge problem'').

The question sits in a tradition nearly as old as the CKM matrix itself \cite{Cabibbo,KM}. Ever since Gatto, Sartori and Tonin, and then Fritzsch, Weinberg, and Wilczek and Zee, tied the Cabibbo angle to a square-root mass ratio \cite{GST,Fritzsch,Weinberg77,WilczekZee}, the recurring hope has been that the mixing matrix is not fundamental data but arithmetic of the mass spectrum. That hope produced GUT-scale mass relations \cite{GeorgiJarlskog}, the texture-zero programme and its systematic surveys \cite{BLM,RRR,LudlGrimus,XingZhao,FritzschXing}, the Froggatt--Nielsen mechanism \cite{FroggattNielsen}, and careful analyses of which texture predictions are general and which are artefacts of a basis choice \cite{HallRasin,BEF}. The present programme belongs to this lineage and differs from it in one decisive respect: conditional on the parent construction's Jordan spread, minimal chains and occupied-node assignments, the mass ratios are predicted rather than refitted for the CKM analysis. The node values and hence every rung angle are therefore fixed before mixing is addressed, and nothing can be re-chosen for the CKM matrix's benefit. For the broader family of division-algebra approaches to standard-model structure, see \cite{GunaydinGursey,Furey,TodorovDV,Boyle,GresnigtCl10}. Experimental values throughout are global-fit results \cite{PDG,CKMfitter,UTfit}.

This paper completes that analysis as far as the present structure permits, and states precisely what remains. The logic and results, tagged \Dtag\ (derived), \Ptag\ (programme input), \Otag\ (open), are:

\begin{enumerate}
\item \textbf{An operator-level mass-modulus audit} (Sec.~\ref{sec:finite}). Let $X_f$ be the sectoral Jordan element with spectral values $(a_f,b_f,c_f)$. The associative spectral algebra generated by this one element defines an ordinary operator $\widehat X_f$ on a three-dimensional spectral space. Compressing the multiplicative lift $\widehat X_f^{\otimes3}|_{\Sym^3}$ onto the three occupied chain nodes gives a linear root-mass operator $\mathcal R_f$ whose singular values are exactly the monomials used in the mass construction; $\mathcal Q_f=\mathcal R_f^\dagger\mathcal R_f$, normalized to its largest eigenvalue, therefore has spectrum $(m_{f1}/m_{f3},m_{f2}/m_{f3},1)$ \Dtag~(conditional on the parent chain and node assignment \Ptag). It embeds in a self-adjoint finite Dirac block. This is a successful operator realization of the predicted mass ratios, but it is not a representation of the full Albert algebra and does not fix a unique Yukawa matrix.
\item \textbf{Quadratic spacetime, cubic particle data, and one finite bridge} (Secs.~\ref{sec:quadraticcubic} and \ref{sec:leavesDF}). The $E_8\supset SU(3)_{\rm geom}\times E_6$ branching separates a geometric factor from the internal $E_6$ acting on the Albert algebra. The Preprints.org scaffolding paper assigns the two geometric $SU(3)$ factors to a split-bioctonionic $(3,3)$ base, two Lorentzian leaves and real four-dimensional internal fibres, while placing the $E_6$ fields on that scaffold \cite{Scaffolding}. We give a further Jordan-algebraic formulation: the base is modelled by the quadratic norm of $J_2(\mathbb H_s)$, the particle sector by the cubic norm of $J_3(\mathbb O_{\mathbb C})$, and their block corner obeys $N_3(\operatorname{diag}(A,\lambda))=\lambda N_2(A)$. This rank-two/rank-three refinement and its proposed soldering interpretation are results and postulates of the present paper, not claims attributed to Ref.~\cite{Scaffolding}. The mass-ratio paper starts before symmetry breaking from two isomorphic copies $J_L\oplus J_R$ of $J_3(\mathbb O_{\mathbb C})$ and, after distinct vacuum choices, identifies them with flavour/charge and square-root-mass frames, respectively \cite{Singh2508}. The proposed $SO(3,3)$ BF/Plebanski dynamics produces two overlapping four-dimensional descendants, with one branch carrying weak-force geometry \cite{WesleySinghIsidroBF}; the emergence paper supplies a compatible operator-level construction \cite{SinghEmergence}. We propose the branchwise identification $(\Sigma_L,J_L)\leftrightarrow$ weak/flavour and $(\Sigma_R,J_R)\leftrightarrow$ square-root mass. The finite Yukawa map connects the two. The norm identities are algebraic facts \Dtag; their physical pairing and the gluing are programme postulates \Ptag, not theorems of the cited constructions.
\item \textbf{A universal chiral channel, and a factorization limit} (Sec.~\ref{sec:universalsoldering}). In the split Albert real form, the full Peirce decomposition is $\mathbf1\oplus\mathbf{16}\oplus\mathbf{10}$ under $\operatorname{Spin}(5,5)$. Choosing $\mathbb H_s\subset\mathbb O_s$ refines $\mathbf{10}$ to the spacetime $\mathbf6_{(3,3)}$ plus a $\mathbf4_{(2,2)}$. Representation theory of the Albert cubic forces this complement to pair the two opposite $\operatorname{Spin}(3,3)$ Weyl modules. This identifies a candidate universal left--right soldering channel \Dtag, conditional on the real-form and spacetime-corner choices \Ptag. If $Y_f=B\otimes\Xi_f$ with common $B$, then $B$ changes only the overall chiral normalization and $V_{\rm CKM}=U_{\Xi_u}^\dagger U_{\Xi_d}$ \Dtag. Consequently the scalar $\lambda_0$ is radial, $B$ can solder chirality, and the family-dependent kernel $\Xi_f$ remains the uncomputed source of mixing \Otag.
\item \textbf{Mass ratios and mixing as two parts of one finite Dirac operator} (Sec.~\ref{sec:finite}). For $H_f=Y_fY_f^\dagger$, the spectrum of $H_f$ contains the squared Yukawa singular values, while its spectral frame is $U_f$; therefore $V_{\rm CKM}=U_u^\dagger U_d$ is the relative frame of the two quark blocks, up to quark rephasings. The exceptional chain data determine the radial or spectral part $\mathcal Q_f$, while the transport and bridge are asked to determine the angular or frame part. This unifies the bookkeeping, but does not turn the latter into a consequence of the former: the finite-Dirac underdetermination proposition below proves that the spectrum alone cannot fix CKM mixing. \Dtag
\item \textbf{An explicit Peirce-changing lift of the Cabibbo bridge} (Sec.~\ref{sec:albertbridge}). Conditional on the adopted identification of the first two physical families with an $F_1$--$F_2$ Peirce pair \Ptag, the local octonionic rotor $g_\chi=\cos\chi\,e_3-\sin\chi\,e_1$ lifts to the Albert inner derivation ${\cal D}_{12}(-g_\chi)=2[L_{p_1-p_2},L_{F_3(-g_\chi)}]$. It changes the chosen Peirce slot while reproducing exactly the conjugate up/anti-down transport and hence $\varphi_{12}=-2\chi$ \Dtag. This constructs the Cabibbo component of the bridge for that embedding; it does not derive or uniquely select the physical family embedding, the edge strength, the physical $23$ and $31$ quark matrix elements, or the right-sector soldering map \Otag.
\item \textbf{Relation to finite noncommutative geometry and the precise novelty} (Sec.~\ref{sec:ncg}). Chamseddine, Connes and Marcolli already classified finite Dirac operators whose moduli include fermion masses and CKM/PMNS data \cite{CCM}; that broad $D_F$ unity is not new. The present proposal instead attempts to select part of those moduli: exceptional-Jordan chains conditionally determine sectoral positive Yukawa moduli, while transport and the weak-to-mass bridge partially constrain their relative left frames. In particle-physics terms this gives proposed individual mass-ratio patterns, cross-sector relations, two conditional CKM moduli, a virtual-node suppression mechanism and a no-go localization of the remaining long-edge dynamics. The construction is still incomplete and does not determine a unique $D_F$.
\item \textbf{A vacuum-selection obstruction} (Appendix~\ref{app:vacuum}). The finite spectral traces displayed in the emergence construction are sectorwise and hence leave $U_u^\dagger U_d$ flat. Under the minimal dynamical promotion of the two Yukawa blocks to bifundamentals, the unique renormalizable orientation-sensitive invariant is $\operatorname{Tr}(X_uX_d)$, whose stationarity condition is $[X_u,X_d]=0$ for nondegenerate spectra. The minimal vacuum therefore aligns the sectors or leaves their relative frame undetermined; independent right-family symmetries make $W_u,W_d$ invisible at every polynomial order. This conditional theorem identifies the additional ingredients required for nine-link sparsity: vacuum-selected projectors and a right-sector locking structure. \Dtag
\item \textbf{A complementary Peirce-edge truncation test} (Appendix~\ref{app:edgepotential}). The minimal radial-quartic cyclic truncation considered for the three lifted edge amplitudes has only equal-magnitude full-rank extrema, or a flat direction when its cubic vanishes. This is not a no-go theorem for a general $C_3$-invariant quartic potential. Coupling the bridge to a charged-sector $\mathbf{27}$ through $\operatorname{Re}d(\Phi,H_q,H_q)$ gives a stable minimum within the assumed subspace $W$, with $\boldsymbol\rho\parallel P_WH_q^\#$ \Dtag. This shows how a charged vacuum can select unequal edges, but the present theory determines neither $W$ nor $P_WH_q^\#$; a general three-component source fits rather than predicts the three real PMNS angles \Otag.
\item \textbf{The lift's three-generation extension: one family excluded, one adopted} (Sec.~\ref{sec:transport}). Extending the two-state textures to eigenvalue-matched zero-diagonal Hamiltonian textures on the full chains fails decisively (non-unitarity of $20$--$69\%$ across every admissible sign assignment); the failure is reported, with code, in Appendix~\ref{app:failed} \Dtag. We adopt the \emph{transport} extension: the ordered product of rung rotations with angles \eqref{eq:lift} along the chain, with mass eigenstates identified with the chain nodes as the ladder normalization dictates. This extension is exactly unitary and preserves the successful two-generation relations. This is a modelling choice supported by the exclusion, not forced by it (general nearest-neighbour textures with free diagonals remain viable); adjacency itself, however, is a consequence of the minimality principle \Dtag.
\item \textbf{$\kappa_{23}$ from the virtual node, under a stated amplitude postulate} (Sec.~\ref{sec:kappa}). The down chain is $a^2b\to abc\to ac^2\to c^3$ with generations at the first, second and \emph{fourth} nodes; defining the effective $2\to3$ angle by the two-step projected transport \emph{amplitude} through the unoccupied $ac^2$ node,
\begin{equation}
\sin\theta^{d}_{23}\;=\;\sin\theta_2\,\sin\theta_3,
\qquad \tan\theta_2=\tfrac{1}{1+\delta},\quad\tan\theta_3=\tfrac{1-\delta}{1+\delta},
\end{equation}
numerically $0.1232$ against the naive $\sqrt{m_s/m_b}=0.1491$, gives at the level of the $\Vcb$ difference form an effective $\kappa_{23}^{\rm eff}=0.633$, i.e.\ $\Vcb=0.0422$, $+0.8\%$ from the 2026 global-fit $0.04188$, with no fitted continuous CKM parameter after the stated structural choices \Dtag~(given the amplitude postulate \Ptag). The up chain has no virtual node and is unsuppressed; the previously unexplained up--down asymmetry of the $(2,3)$ sector is a counting statement about chain nodes. Two fragilities are disclosed in Sec.~\ref{sec:kappa}: alternative re-unitarization conventions for the composed transport spread $\Vcb$ widely, and the result is strongly sensitive to the parent theory's $\sqrt{m_c/m_t}$.
\item \textbf{The phase sector is fixed by one rule and one no-go} (Sec.~\ref{sec:phases}). Phases are inserted only where the conjugation theorem creates them: on the shared Cabibbo edge, with the balanced value $\varphi_{12}=90^\circ$ \Ptag. A relative phase between the distinct $(2,3)$ edges ($C$ for down, $B$ for up) is \emph{excluded}: the scan of Lemma~\ref{lem:nogo} shows any such phase repairs $\Vub$ only at the cost of destroying $\Vcb$ \Dtag.
\item \textbf{The conditional no-fit layer, and its one deficit} (Sec.~\ref{sec:assembly}). Given the parent node assignments, transport reading, virtual-node amplitude rule, real $(2,3)$ blocks and balanced shared-edge rotor, with $\delta^2=3/8$ and no fitted continuous CKM parameters: $\Vus=0.2371$ $(+5.3\%)$, $\Vcb=0.0422$ $(+0.8\%)$, $\Vts=0.0411$ $(-0.2\%)$, $\Vtd=0.0099$ $(+14.5\%)$, and the exact Fritzsch ratio $\Vub/\Vcb=\sqrt{m_u/m_c}=0.0424$, a factor $2.1$ below the measured $0.090$. The long-standing Fritzsch shortfall is thus reproduced, localized, and shown to be irreparable within the minimal transport model.
\item \textbf{The remaining two degrees of freedom are fitted, and named} (Sec.~\ref{sec:bridge}). A unitary CKM matrix has exactly four physical parameters. Layer 1 conditionally fixes two of them without fitted continuous CKM parameters ($\theta_{12}$, $\theta_{23}$); the remaining two ($\theta_{13}$, $\delta_{CP}$) are fitted through the minimal long-edge insertion $R_{13}(\varepsilon,\omega)$.  The two balanced-rotor orientations give discrete branches: the displayed branch has $(\varepsilon,\omega)=(0.002079,288.2^\circ)$, while the opposite orientation has $(0.005358,202.2^\circ)$.  Both reproduce $\Vub$ and $\delta_{CP}$. All other CKM quantities ($J$, $\beta$, $\gamma$, $\bar\rho$, $\bar\eta$, $\Vtd$, $\Vts$, $A$) are then functions of the four numbers, and Table~\ref{tab:layer2} confirms internal consistency at the few-percent level; they are not additional predictions.  The fitted bridge is therefore not unique until its discrete orientation is derived.  In particular, $\varepsilon/(s^u_{12}s^u_{23})$ is $0.604$ on the displayed branch but $1.557$ on the opposite branch, so the former value is not a robust target.
\end{enumerate}

Throughout, $\delta=\sqrt{3/8}$ is the universal Jordan spread \cite{Singh2508}; every mass-ratio input is a theoretical ratio rather than a refitted one, so the Layer-1 mass-to-mixing residuals are genuine theory--experiment comparisons. Section~\ref{sec:afhm} uses the finite-Dirac audit and the explicit Albert lift to make the comparison with the nine-link texture scan of Arkani-Hamed, Figueiredo, Hall and Manzari precise. Appendix~\ref{app:vacuum} proves the minimal bifundamental vacuum obstruction; Appendix~\ref{app:edgepotential} gives the radial-quartic truncation result and sourced completion; Appendix~\ref{app:proofs} collects other proofs and closed forms; Appendix~\ref{app:numerics} documents the numerical checks; Appendix~\ref{app:failed} records the failed alternative.

\section{Inputs: chains, monomials, and the adjacent-edge lift}
\label{sec:inputs}

The minimal chains and square-root masses (per-sector units) are \cite{Singh2508,TeliSinghMass}
\begin{align}
\text{down}&:\quad a^2b\ \xrightarrow{E}\ abc\ \xrightarrow{C}\ ac^2\ \xrightarrow{E}\ c^3,
&(d,s,b)\ \text{at nodes}\ (1,2,4),\ \ ac^2\ \text{unoccupied},\\
\text{up}&:\quad a^2b\ \xrightarrow{E}\ abc\ \xrightarrow{B}\ b^2c,
&(u,c,t)\ \text{at nodes}\ (1,2,3),
\end{align}
with $(a,b,c)=(s_F-\delta,\ s_F,\ s_F+\delta)$, $s_d=1$, $s_u=\tfrac23$, $\delta=\sqrt{3/8}$. The node values (square-root masses; the virtual value in brackets) are
\begin{equation}
\begin{aligned}
v^{d}&=\big((1-\delta)^2,\ \ 1-\delta^2,\ \ [\,(1-\delta)(1+\delta)^2\,],\ \ (1+\delta)^3\big),\\
v^{u}&=\big(\tfrac23(\tfrac23-\delta)^2,\ \ (\tfrac23-\delta)\tfrac23(\tfrac23+\delta),\ \ \tfrac49(\tfrac23+\delta)\big),
\end{aligned}
\end{equation}
where for the up chain we list $(a^2b,\,abc,\,b^2c)$ evaluated at $s_u=\tfrac23$. The \emph{rung tangents}, the adjacent node ratios (exactly the square-root mass ratios where both nodes are occupied), are
\begin{equation}
\begin{aligned}
t^{d}_{\rm rung}&=\Big(\tfrac{1-\delta}{1+\delta},\ \tfrac{1}{1+\delta},\ \tfrac{1-\delta}{1+\delta}\Big)=(0.2404,\ 0.6202,\ 0.2404),\\
t^{u}_{\rm rung}&=\Big(\tfrac{2/3-\delta}{2/3+\delta},\ \tfrac{2/3-\delta}{2/3}\Big)=(0.0424,\ 0.0814).
\end{aligned}
\label{eq:rungs}
\end{equation}
In words: each node of a chain carries the square-root mass of one generation, each rung carries the mixing angle of the two generations it joins, and the whole quark flavour sector is to be read off the geometry of two short graphs. The adjacent-edge lift \eqref{eq:lift} is exact for the two-state texture
$M_{ij}=\big(\begin{smallmatrix}0&\sqrt{m_im_j}\\ \sqrt{m_im_j}&m_j-m_i\end{smallmatrix}\big)$ \cite{TeliSinghMixing}, in the lineage of the classic square-root relations \cite{Fritzsch,GST} and of the earlier CKM analysis of the programme \cite{PatelSingh}; its three-generation extension is the subject of the next section. Experimental targets are derived coherently from the PDG 2026 global-fit Wolfenstein central values \cite{PDG}, quoted in Tables~\ref{tab:layer1}--\ref{tab:layer2}.  For the numerical comparison we adopt the common parent-theory matching baseline $\mu_*=M_Z$; the theoretical rung inputs involving running mass ratios are to be interpreted at that same baseline, not mixed across scales.

\section{From the Jordan spectrum to a finite Dirac mass modulus}
\label{sec:finite}

The node prescription above is a spectral statement. To use it in a standard-model Yukawa operator one must show that the selected monomials are the singular values of a genuine linear map, rather than only labels on a diagram. This section performs that audit. It also separates what the audit proves from the chiral information that it cannot supply.

Let $X_f$ be a Hermitian element on the real spectral slice of the complex Albert algebra $J_3(\mathbb O_{\mathbb C})$ used in the parent construction (equivalently, an element of the formally real $J_3(\mathbb O)$ slice for the present eigenvalue problem), with spectral resolution
\begin{equation}
X_f=a_fP_a+b_fP_b+c_fP_c,
\qquad P_i\circ P_j=\delta_{ij}P_i,
\qquad P_a+P_b+P_c=\mathbf1.
\label{eq:jordanspec}
\end{equation}
The exceptional algebra as a whole has no faithful realization as the self-adjoint part of an associative matrix algebra. No such realization is needed here. The Jordan subalgebra generated by one element is power-associative, and its three spectral idempotents can be mapped to the ordinary projections $|a\rangle\langle a|$, $|b\rangle\langle b|$, $|c\rangle\langle c|$ on a spectral space $V_f\simeq\mathbb C^3$. Thus
\begin{equation}
\widehat X_f=a_f|a\rangle\langle a|+b_f|b\rangle\langle b|+c_f|c\rangle\langle c|
\label{eq:Xhat}
\end{equation}
is an ordinary positive operator for the charged sectors considered here.

On $\Sym^3(V_f)$ use normalized occupation states $|pqr\rangle$, $p+q+r=3$. Define the multiplicative symmetric-cube lift
\begin{equation}
\Gamma_3(\widehat X_f)
\equiv
\left.\widehat X_f^{\otimes3}\right|_{\Sym^3(V_f)}.
\label{eq:gamma3}
\end{equation}
It acts diagonally,
\begin{equation}
\Gamma_3(\widehat X_f)|pqr\rangle
=a_f^pb_f^qc_f^r|pqr\rangle .
\label{eq:monomialspec}
\end{equation}
The familiar square-root Clebsch coefficients occur in the action of the $SU(3)$ ladder generators between normalized occupation states. They do not multiply the diagonal eigenvalues in \eqref{eq:monomialspec}. This is why the node monomials can be audited independently of the rung amplitudes.

\begin{proposition}[Occupied-chain root-mass operator]
\label{prop:rootoperator}
Let $P_{\rm occ}^{f}$ project onto the three occupied nodes of the sectoral chain and define
\begin{equation}
\mathcal R_f
=
\left.P_{\rm occ}^{f}\Gamma_3(\widehat X_f)P_{\rm occ}^{f}
\right|_{P_{\rm occ}^{f}\Sym^3(V_f)}.
\label{eq:Rf}
\end{equation}
In the generation ordering used in this paper,
\begin{align}
P_{\rm occ}^{d}&=|210\rangle\langle210|+|111\rangle\langle111|+|003\rangle\langle003|,
&\mathcal R_d&=\operatorname{diag}(a_d^2b_d,a_db_dc_d,c_d^3),\nonumber\\
P_{\rm occ}^{u}&=|210\rangle\langle210|+|111\rangle\langle111|+|021\rangle\langle021|,
&\mathcal R_u&=\operatorname{diag}(a_u^2b_u,a_ub_uc_u,b_u^2c_u).
\label{eq:Rud}
\end{align}
The omitted down state $|102\rangle$ is precisely the virtual $ac^2$ node. If $\rho_{f3}$ denotes the largest occupied eigenvalue and
\begin{equation}
\mathcal Q_f=\frac{\mathcal R_f^\dagger\mathcal R_f}{\rho_{f3}^2},
\label{eq:Qf}
\end{equation}
then
\begin{equation}
\operatorname{spec}\mathcal Q_f
=\left(\frac{m_{f1}}{m_{f3}},\frac{m_{f2}}{m_{f3}},1\right)
\label{eq:Qspec}
\end{equation}
in exactly the root-mass convention of the parent construction. \Dtag~(conditional on its chain and occupied-node assignment \Ptag)
\end{proposition}

\begin{proof}
Equation~\eqref{eq:monomialspec} makes both compressed operators in \eqref{eq:Rud} diagonal. Their singular values are therefore the absolute values of the three occupied monomials, which are positive for $s_d=1$, $s_u=2/3$ and $\delta=\sqrt{3/8}$. The parent prescription identifies these monomials with quantities proportional to $\sqrt{m_{fi}}$. Squaring and normalizing to the heaviest occupied node gives \eqref{eq:Qspec}. No additional singular-value assumption is required.
\end{proof}

This proposition is narrower than the assertion that the exceptional Jordan matrix itself is a Hilbert-space Dirac operator. It represents only the associative spectral algebra generated by $X_f$ and then applies the ordinary tensor functor. It also shows why the Jordan quadratic representation $U_X$ is not by itself the whole answer: $U_X(P_i)=x_i^2P_i$ squares the three original Jordan eigenvalues, whereas the physical construction first forms selected degree-3 monomials and only then squares them through $\mathcal R_f^\dagger\mathcal R_f$.

There is a second limitation. The projector $P_{\rm occ}^f$ is not determined by the unordered eigenvalues of $X_f$ alone: it contains the minimal-chain and generation assignment of the parent mass construction. Equation~\eqref{eq:Rf} is therefore canonical only relative to the chosen Jordan frame and occupied chain, with the projector co-transported under a change of frame. Proposition~\ref{prop:rootoperator} validates the operator realization of that prescription; it does not independently derive the prescription or promote it to an automorphism-invariant function of $X_f$.

\subsection{Yukawa family and finite Dirac block}

Let $U_f$ be the left-handed frame supplied by the transport construction of Sec.~\ref{sec:transport}, let $W_f\in U(3)$ be the corresponding right-handed frame, and let $y_{f3}(\mu_*)$ be the heaviest Yukawa coupling of the sector at a common matching scale $\mu_*$. The general Yukawa matrix compatible with Proposition~\ref{prop:rootoperator} is its singular-value decomposition,
\begin{equation}
\boxed{\quad Y_f(\mu_*)=y_{f3}(\mu_*)\,U_f\mathcal Q_fW_f^\dagger .\quad}
\label{eq:Ysvd}
\end{equation}
The sectoral finite Dirac block is then
\begin{equation}
\mathcal D_F^{(f)}=
\begin{pmatrix}
0&Y_f\\ Y_f^\dagger&0
\end{pmatrix},
\qquad
(\mathcal D_F^{(f)})^2=
\begin{pmatrix}
Y_fY_f^\dagger&0\\0&Y_f^\dagger Y_f
\end{pmatrix},
\label{eq:finiteD}
\end{equation}
with spectrum
\begin{equation}
\operatorname{spec}\mathcal D_F^{(f)}
=\left\{\pm y_{f3}(\mu_*)\frac{m_{f1}}{m_{f3}},
\pm y_{f3}(\mu_*)\frac{m_{f2}}{m_{f3}},
\pm y_{f3}(\mu_*)\right\}.
\label{eq:Dspec}
\end{equation}
Multiplication by $v(\mu_*)/\sqrt2$ converts this dimensionless Yukawa block into the running mass block. Equations~\eqref{eq:Ysvd}--\eqref{eq:Dspec} are matching conditions in one renormalization scheme at $\mu_*$, not renormalization-group invariants; comparison with masses quoted at another scale requires standard-model running and threshold matching. Particle--antiparticle doubling, colour multiplicities, the real structure and possible Majorana terms can be added in the usual finite spectral data; they do not alter the charged-sector singular-value statement \cite{CCM}.

\subsection{Placement in the fluctuated Dirac operator}
\label{sec:fulloperator}

The emergence analysis of the programme separates four operator contributions rather than identifying every octonionic datum with the same $D_F$ \cite{SinghEmergence}.  The related $SO(3,3)$ BF construction supplies the proposed symmetry-breaking and two-leaf dynamics, but does not by itself determine the finite family kernel \cite{WesleySinghIsidroBF}. A compatible pre-reduction gauge completion may be written schematically as
\begin{equation}
\mathcal D_{A,6}
=\Gamma^A\!\left(\nabla_A\otimes\mathbf1
+i\,\mathbf1\otimes\mathcal A_A^{\rm oct}\right)
+\Gamma_7\otimes\mathcal D_{F,\Phi},
\qquad
\mathcal D_{F,\Phi}
=D_F+\Phi_{LR}+J_F\Phi_{LR}J_F^{-1}.
\label{eq:Dgauge6}
\end{equation}
Here the split-biquaternionic directions supply the $(3,3)$ kinetic operator, the non-scalar octonionic bosonic directions supply internal vector connections including $SU(3)_c$, and the scalar/bifermionic channel supplies the left--right Higgs bridge. After localization and BF reduction, the corresponding observed-leaf operator is
\begin{equation}
D_A^{\rm phys}
=D_\Sigma^{\rm phys}\otimes\mathbf1
+A_\Sigma^{\rm phys}
+\gamma_5\otimes\mathcal D_{F,\Phi}^{\rm phys}.
\label{eq:Dleaf}
\end{equation}
The finite block constructed in Eqs.~\eqref{eq:Ysvd}--\eqref{eq:finiteD} is the dimensionless Yukawa datum $\mathcal D_F^Y$. We define its broken-vacuum normalization by the matching convention
\begin{equation}
\left.\mathcal D_{F,\Phi}^{\rm phys}\right|_{\langle H\rangle}
=\frac{v(\mu_*)}{\sqrt2}\,\mathcal D_F^Y(\mu_*).
\label{eq:vacuumDF}
\end{equation}
Equation~\eqref{eq:vacuumDF} is not a derivation of $v$ or of the Higgs vacuum. It fixes conventions: $v/\sqrt2$ belongs either to the vacuum value of the fluctuated operator or to the displayed mass block, but not to both. The simpler product formula $D_{3,3}\otimes\mathbf1+\Gamma_7\otimes(v/\sqrt2)\mathcal D_F^Y$ is the gauge-free broken-vacuum specialization of \eqref{eq:Dgauge6}, not a competing Dirac operator.

The two octonionic roles do not double count one another. On a quark block the family modulus acts as $\mathcal Q_f\otimes\mathbf1_3^{\rm colour}$, whereas the colour connection contains $\mathbf1_{\rm family}\otimes T_c^a$; hence
\begin{equation}
\big[\mathcal Q_f\otimes\mathbf1_3^{\rm colour},
\mathbf1_{\rm family}\otimes T_c^a\big]=0.
\label{eq:colourcommute}
\end{equation}
The octonionic colour fluctuation is therefore gauge data, while the exceptional-Jordan construction supplies family spectral data. Equation~\eqref{eq:colourcommute} is a consistency check of the observed-leaf tensor factorization; it does not by itself derive that factorization from the pre-breaking octonionic theory.

This result refines the emergence paper's deliberately minimal factorization $Y_f=y_fY_J$. The common schematic $Y_J$ is replaced, sector by sector, by
\begin{equation}
Y_f=y_{f3}U_f\mathcal Q_fW_f^\dagger.
\label{eq:YJrefine}
\end{equation}
The replacement changes neither the observed gauge traces nor the colour multiplicities. It does refine the charged-sector finite invariants entering the heat-kernel coefficients to
\begin{equation}
a_{\rm ch}=\sum_{f=e,u,d}n_f|y_{f3}|^2\operatorname{Tr}\mathcal Q_f^2,
\qquad
b_{\rm ch}=\sum_{f=e,u,d}n_f|y_{f3}|^4\operatorname{Tr}\mathcal Q_f^4,
\qquad (n_e,n_u,n_d)=(1,3,3),
\label{eq:abrefine}
\end{equation}
with the neutrino and possible Majorana contributions added separately once their moduli are specified. The unknown frames cancel from these traces. Consequently the spectral action needs exactly the modulus that the occupied-chain construction supplies, whereas CKM mixing needs the relative left frames that it does not.

Equations~\eqref{eq:Dgauge6}--\eqref{eq:abrefine} are an operator dictionary, not a derivation of the full product geometry from the pre-geometric dynamics. Nor do they claim that a complete exceptional spectral triple has been constructed; genuinely Jordan-geometric formulations require modified axioms \cite{BoyleFarnsworth}.

\subsection{Why quadratic spacetime and cubic particle geometry?}
\label{sec:quadraticcubic}

Because two different Jordan algebras enter the construction, it is useful to pause and explain their distinct roles.  The kinematic-scaffolding paper separates the two geometric $SU(3)$ factors from the internal $E_6$ factors and constructs a $(3,3)$ base together with four-dimensional internal fibres \cite{Scaffolding}.  It does not formulate this as the rank-two/rank-three Jordan decomposition below; that is the refinement proposed here. For either chiral copy, the exceptional branching used in the parent construction has the schematic complex form \cite{SinghEmergence}
\begin{equation}
\mathfrak e_8
= (\mathfrak{su}(3)_{\rm geom},\mathbf1)
\oplus(\mathbf1,\mathfrak e_6)
\oplus(\mathbf3,\mathbf{27})
\oplus(\overline{\mathbf3},\overline{\mathbf{27}}),
\label{eq:e8branchped}
\end{equation}
with the appropriate real forms understood. The first two summands commute, while the mixed summands link them. The point is therefore not that one $SU(3)$ directly generates both spacetime and particles. Rather, $SU(3)_{\rm geom}$ labels the proposed geometric branch, whereas the $\mathbf{27}$ is naturally modelled by the complex Albert algebra $J_3(\mathbb O_{\mathbb C})$ on which $E_6$ acts. The familiar internal subgroup $SU(3)\times SU(3)\times SU(3)\subset E_6$ is a resolution of this Albert-algebra sector, not a second copy of the geometric $SU(3)$.

On the spacetime side, write an element of the six-dimensional Jordan algebra of $2\times2$ Hermitian split-quaternionic matrices as
\begin{equation}
A=\begin{pmatrix}\alpha&q\\ \bar q&\beta\end{pmatrix}
\in J_2(\mathbb H_s),
\qquad
N_2(A)=\det A=\alpha\beta-N(q).
\label{eq:J2norm}
\end{equation}
With $t=(\alpha+\beta)/2$ and $z=(\alpha-\beta)/2$,
\begin{equation}
N_2(A)=t^2-z^2-N(q).
\label{eq:J2signature}
\end{equation}
The split-quaternion norm $N(q)$ has signature $(2,2)$, so $N_2$ has signature $(3,3)$. Its connected determinant-preserving group is $SO_0(3,3)$, with spin cover
\begin{equation}
\operatorname{Spin}(3,3)\simeq SL(4,\mathbb R).
\label{eq:spin33}
\end{equation}
Thus $J_2(\mathbb H_s)$ supplies more than a six-dimensional vector space: its quadratic Jordan determinant supplies the interval and null cone of the proposed six-dimensional base.

The route from the doubled geometric branching to this real form is a programme identification, not a consequence of the compact groups alone \cite{Scaffolding}. Each geometric factor is reduced as
\begin{equation}
SU(3)_{\rm geom}^{L,R}\longrightarrow
SU(2)_{\rm geom}^{L,R}\times U(1)_{\rm geom}^{L,R},
\label{eq:geombreak}
\end{equation}
and the two branches are organized by the split-complex unit $\omega^2=+1$, or equivalently by the idempotents $e_\pm=(1\pm\omega)/2$. This provides a canonical grading with which the two branch contributions can be assigned opposite signs. A qualification is important: the compact rotations act through $SO(3)_L\times SO(3)_R$, the maximal compact subgroup of $SO(3,3)$, but they do not generate the noncompact boosts or prove the split metric. The $(3,3)$ statement rests on the explicit quadratic form \eqref{eq:J2signature}; $\omega$ supplies a natural place for the relative sign, while the soldering of the two rank-three bundles to the tangent space remains a programme input.

On the particle side, a general element $X\in J_3(\mathbb O_{\mathbb C})$ carries the cubic Jordan determinant
\begin{equation}
N_3(X)=\det X,
\label{eq:J3norm}
\end{equation}
whose reduced structure group is $E_6(\mathbb C)$ and whose automorphism group is $F_4(\mathbb C)$, up to the usual connectedness and finite-centre qualifications \cite{SpringerVeldkamp}. The invariant map $\mathbf{27}^{\otimes3}\to\mathbf1$ is the representation-theoretic expression of this cubic norm. In the present flavour application, the three Jordan eigenvalues, their degree-three monomials and the occupied $\Sym^3(\mathbf3)$ chains are particle spectral data. They are not spacetime coordinates.

The contrast can be summarized as follows.
\begin{center}
\small
\begin{tabular}{>{\raggedright\arraybackslash}p{0.22\textwidth}>{\raggedright\arraybackslash}p{0.31\textwidth}>{\raggedright\arraybackslash}p{0.31\textwidth}}
\toprule
& spacetime & particles \\
\midrule
Jordan algebra & $J_2(\mathbb H_s)$ & $J_3(\mathbb O_{\mathbb C})$ \\
rank & $2$ & $3$ \\
invariant & quadratic $N_2$ & cubic $N_3$ \\
structure group & $SO_0(3,3)$ & $E_6(\mathbb C)$ \\
proposed role & $(3,3)$ interval, null cone and causal base & $\mathbf{27}$ spectrum, charges and invariant couplings \\
\bottomrule
\end{tabular}
\end{center}
In short, rank two governs the proposed spacetime geometry, while rank three governs the particle sector.

There is also an exact relation between the two invariants. After choosing an embedding $\mathbb H_s\subset\mathbb O_{\mathbb C}$ and a primitive idempotent that selects a $2+1$ block, one may embed $A\in J_2(\mathbb H_s)$ as
\begin{equation}
X=\begin{pmatrix}A&0\\0&\lambda_0\end{pmatrix}
\in J_3(\mathbb O_{\mathbb C}).
\label{eq:cornerembed}
\end{equation}
The Albert determinant then factorizes exactly:
\begin{equation}
\boxed{N_3(X)=\lambda_0 N_2(A).}
\label{eq:cornerfactor}
\end{equation}
For fixed nonzero $\lambda_0$, the cubic norm restricts to the quadratic spacetime norm. The scalar $\lambda_0$ rescales this corner norm but does not rotate its Peirce frame; it is therefore a radial or matching-scale datum, not a CKM or PMNS mixing variable. This motivates the phrase that spacetime is a rank-two Peirce corner of the rank-three matter geometry. The relevant corner belongs to the split Albert real form $J_3(\mathbb O_s)$; it does not lie in the formally real $J_3(\mathbb O)$, although both real forms sit inside $J_3(\mathbb O_{\mathbb C})$. The corner relation is therefore not a unique dynamical reduction: the embedding of $\mathbb H_s$, the primitive idempotent and the physical real form must all be selected. This distinction is directly relevant to the CKM problem. The same missing vacuum-selected projectors that obstruct a preferred chiral texture may also be the data needed to make the quadratic corner of the cubic geometry physical. The algebraic factorization \eqref{eq:cornerfactor} is exact; identifying its projectors with the family and spacetime projectors of the broken vacuum remains an open programme target \Otag.

\subsection{The full Peirce complement and a universal chiral-soldering channel}
\label{sec:universalsoldering}

The block-diagonal restriction \eqref{eq:cornerembed} suppresses the part of the Albert algebra that is relevant to chirality.  Let $p$ be the primitive idempotent selecting the displayed corner.  In the split real form,
\begin{equation}
 J_3(\mathbb O_s)
 =\mathbb Rp\oplus J_{1/2}(p)\oplus J_0(p)
 \simeq \mathbf1\oplus\mathbf{16}\oplus\mathbf{10}
 \label{eq:Peirce1510}
\end{equation}
under the corresponding $\operatorname{Spin}(5,5)$ subgroup of $E_{6(6)}$.  Write
\begin{equation}
 X=\lambda_0p+\Psi+V,
 \qquad \Psi\in J_{1/2}(p),\qquad
 V\in J_0(p)\simeq J_2(\mathbb O_s).
 \label{eq:PeircefullX}
\end{equation}
Up to conventional signs and normalizations, the cubic invariant has the standard form
\begin{equation}
 N_3(X)=\lambda_0 q_{5,5}(V)+\langle\Psi,\Gamma(V)\Psi\rangle .
 \label{eq:PeircecubicSpin55}
\end{equation}
Thus the scalar multiplies the vector norm, while the ten-vector acts by Clifford multiplication on the chiral spinor.  Equation~\eqref{eq:PeircecubicSpin55}, rather than the scalar corner identity alone, contains the possible soldering information.

Choose a compatible split-quaternion subalgebra and a Cayley--Dickson unit $\ell$,
\begin{equation}
 \mathbb O_s=\mathbb H_s\oplus\mathbb H_s\ell .
 \label{eq:splitOH}
\end{equation}
The Peirce ten-vector then decomposes orthogonally as
\begin{equation}
 \boxed{
 J_2(\mathbb O_s)
 =J_2(\mathbb H_s)\oplus\mathbb H_s\ell,
 \qquad
 \mathbf{10}
 =(\mathbf6,\mathbf1,\mathbf1)
 \oplus(\mathbf1,\mathbf2,\mathbf2) .}
 \label{eq:J2Osdecomp}
\end{equation}
Here the subgroup and spinor branching are
\begin{align}
 \operatorname{Spin}(5,5)&\supset
 \operatorname{Spin}(3,3)\times\operatorname{Spin}(2,2)
 \simeq SL(4,\mathbb R)\times SL(2,\mathbb R)_+\times SL(2,\mathbb R)_-,
 \label{eq:Spin55branch}\\
 \mathbf{16}&=(\mathbf4,\mathbf2,\mathbf1)
 \oplus(\mathbf4^*,\mathbf1,\mathbf2).
 \label{eq:spinor16branch}
\end{align}
The determinant on $J_2(\mathbb H_s)$ has signature $(3,3)$, while the complementary split-quaternion norm has signature $(2,2)$; together they give the $(5,5)$ vector norm.  The dimensions parallel the six-dimensional base and four-dimensional internal fibre of the scaffolding construction \cite{Scaffolding}, but the $J_2(\mathbb O_s)$ embedding, the $(2,2)$ complementary signature and the $\operatorname{Spin}(5,5)$ branching are additional algebraic statements made here.

Write $V=A+B$, with $A\in J_2(\mathbb H_s)$ and $B\in\mathbb H_s\ell$, and write the two summands of $\Psi$ as $\psi^{a\alpha}$ and $\chi_a{}^{\dot\alpha}$.  The subgroup representations force Eq.~\eqref{eq:PeircecubicSpin55} to have the invariant structure
\begin{align}
 N_3(X)={}&\lambda_0\bigl(q_{3,3}(A)+q_{2,2}(B)\bigr)
 +c_1 A_{ab}\epsilon_{\alpha\beta}\psi^{a\alpha}\psi^{b\beta}
 +c_2\widetilde A^{ab}\epsilon_{\dot\alpha\dot\beta}
       \chi_a{}^{\dot\alpha}\chi_b{}^{\dot\beta}
 \notag\\
 &+c_3B_{\alpha\dot\beta}\psi^{a\alpha}\chi_a{}^{\dot\beta},
 \label{eq:Albertchiralbranch}
\end{align}
where the nonzero $c_i$ depend on normalization and the relative signs are convention-dependent.  The first two spinor terms are the $\operatorname{Spin}(3,3)$ vector--Weyl pairings generated by $A$.  The last term is a $\operatorname{Spin}(3,3)$ scalar that pairs its opposite Weyl modules.  Consequently
\begin{equation}
 B\in(\mathbf1,\mathbf2,\mathbf2)
 \label{eq:Bchannel}
\end{equation}
has the representation character of a universal chiral-soldering channel.  This conclusion is algebraic.  Interpreting $B$ as a physical four-dimensional Higgs or finite Dirac field additionally requires the real-form selection, the localization from $\operatorname{Spin}(3,3)$ to the observed Lorentz group, and an invariant coupling to the programme's physical Clifford-module fermions.  Neither the kinematic base-and-fibre construction nor the BF reduction establishes those finite-operator steps \cite{Scaffolding,WesleySinghIsidroBF}.

The separation also gives a useful conditional factorization theorem.  Suppose that, after localization and projection onto a common nonzero chiral mode, the sectoral finite bridges have the product form
\begin{equation}
 Y_f=B\otimes\Xi_f,\qquad f=u,d,
 \label{eq:YBXi}
\end{equation}
where the same $B$ acts on the universal chiral factor and $\Xi_f$ acts on family space.  Then
\begin{equation}
 Y_fY_f^\dagger=(BB^\dagger)\otimes(\Xi_f\Xi_f^\dagger).
 \label{eq:HfactorBXi}
\end{equation}
Projection onto any common eigenmode $s$ of $BB^\dagger$ with eigenvalue $b_s>0$ gives
\begin{equation}
 H_f^{\rm family}=b_s\,\Xi_f\Xi_f^\dagger,
 \qquad
 \boxed{V_{\rm CKM}=U_{\Xi_u}^\dagger U_{\Xi_d}}.
 \label{eq:CKMXifactor}
\end{equation}
The universal factor fixes an overall chiral normalization but not the family eigenvectors.  Therefore $B$ can solve a more primitive left--right soldering problem without determining CKM mixing.  The scalar $\lambda_0$ is still more limited: it sets a radial normalization in Eqs.~\eqref{eq:PeircecubicSpin55} and \eqref{eq:Albertchiralbranch}, but carries no orientation.  Nontrivial mixing requires the family kernel $\Xi_f$, equivalently vacuum-selected Peirce projectors and edge matrix elements.  Equations~\eqref{eq:YBXi}--\eqref{eq:CKMXifactor} are conditional: sector-dependent $B_f$, entanglement of chiral and family indices, or different physical-mode projections would invalidate the cancellation and would themselves constitute additional bridge dynamics.

Finally, the two geometric $SU(3)$ factors do not by themselves generate the split Albert automorphism group.  At the complex level,
\begin{equation}
 \mathfrak f_4^{\mathbb C}
 =(\mathbf8,\mathbf1)\oplus(\mathbf1,\mathbf8)
 \oplus(\mathbf3,\overline{\mathbf6})
 \oplus(\overline{\mathbf3},\mathbf6),
 \label{eq:F4A2A2}
\end{equation}
so an $A_2+A_2$ core needs another $36$ mixed generators.  In the present scaffold the two geometric factors lie in different $\mathfrak e_{8L}\oplus\mathfrak e_{8R}$ summands and have vanishing cross-bracket \cite{Scaffolding}.  Moreover, $B$ in Eq.~\eqref{eq:Bchannel} is a coordinate in the Peirce ten-vector, not one of the missing adjoint generators.  Completing the two geometric factors to a common $F_{4(4)}$ would therefore be an additional cross-sector gluing hypothesis, not a consequence of the existing branching.
There is also a real-form mismatch to resolve: the split completion uses an
$\mathfrak{sl}(3,\mathbb R)\oplus\mathfrak{sl}(3,\mathbb R)$ core, whereas the
displayed geometric factors are compact $\mathfrak{su}(3)$ algebras.  A real-form
selection or analytic continuation would therefore be required even after the
missing mixed transformations were supplied.

\subsection{Two overlapping leaves and two Jordan copies}
\label{sec:leavesDF}

The mass-ratio construction begins in the left--right symmetric phase with two isomorphic complex Albert algebras,
\begin{equation}
J_L\oplus J_R
=J_3(\mathbb O_{\mathbb C})_L\oplus J_3(\mathbb O_{\mathbb C})_R,
\qquad
E_{6L}\curvearrowright J_L,
\qquad E_{6R}\curvearrowright J_R,
\label{eq:twoAlbert}
\end{equation}
where the two $E_6$ factors act as the corresponding determinant-preserving structure groups, with the physical real form and covering understood, and no preferred identification of charge and mass directions exists before triality and left--right breaking \cite{Singh2508}. Distinct post-breaking octonionic vacuum choices then give $J_L$ the flavour/charge frame and $J_R$ the square-root-mass frame. In the notation of Ref.~\cite{Singh2508}, the weak basis on the left is not aligned with the Jordan mass basis on the right; this is the stated origin of fermion mixing.

The spacetime construction supplies a parallel doubling. The split-biquaternionic $(3,3)$ base is proposed to yield two overlapping Lorentzian leaves in the $SO(3,3)$ BF/Plebanski construction \cite{WesleySinghIsidroBF}, with a compatible operator-level treatment in Ref.~\cite{SinghEmergence},
\begin{equation}
\Sigma_R=\operatorname{Im}\mathbb H_R\oplus\operatorname{span}\{t_L\},
\qquad
\Sigma_L=\operatorname{Im}\mathbb H_L\oplus\operatorname{span}\{\omega t_R\},
\label{eq:twoleaves}
\end{equation}
with $\Sigma_L$ carrying the weak-force geometry and $\Sigma_R$ the complementary gravity/right branch after BF reduction. The natural synthesis is the branchwise pairing
\begin{equation}
(\Sigma_L,J_L)\longleftrightarrow \text{weak/flavour frame},
\qquad
(\Sigma_R,J_R)\longleftrightarrow \text{square-root-mass frame},
\label{eq:branchpairing}
\end{equation}
and, for each fermion sector, a finite bridge
\begin{equation}
Y_f:\mathcal H_{R,f}^{(\Sigma_R,J_R)}\longrightarrow
\mathcal H_{L,f}^{(\Sigma_L,J_L)}.
\label{eq:leafbridge}
\end{equation}
This brings the geometric and mass-ratio statements into one diagram, but its status must be kept exact: Eqs.~\eqref{eq:twoAlbert} and \eqref{eq:twoleaves} belong to the cited constructions, whereas the cross-identification \eqref{eq:branchpairing} and the claim that the microscopic bridge is induced by gluing on $\Sigma_L\cap\Sigma_R$ are programme postulates \Ptag. No leaf-overlap kernel has yet been calculated.

Moreover, CKM is not simply the left--right overlap matrix. Each $Y_f$ first induces a left mass frame through $H_f=Y_fY_f^\dagger$ on the common weak doublet space; CKM is then the relative frame of the \emph{up and down} left blocks. The two-leaf picture proposes a geometric origin for these induced frames, but does not bypass this standard chiral structure.

\subsection{An explicit Albert-algebra lift of the Cabibbo bridge}
\label{sec:albertbridge}

The preceding discussion identifies a missing finite map.  The local Cabibbo calculation of Ref.~\cite{TeliSinghCP} supplies more: its octonionic rotor admits an exact lift that changes the Albert Peirce slot as well as the octonionic coefficient.  We use the standard coordinates
\begin{equation}
 X=\begin{pmatrix}
 \xi_1&x_3&\bar x_2\\
 \bar x_3&\xi_2&x_1\\
 x_2&\bar x_1&\xi_3
 \end{pmatrix},
 \qquad
 p_i=E_{ii},
 \qquad
 V_{23}=F_1(\mathbb O_{\mathbb C}),\quad
 V_{31}=F_2(\mathbb O_{\mathbb C}),\quad
 V_{12}=F_3(\mathbb O_{\mathbb C}),
 \label{eq:AlbertPeirce}
\end{equation}
where $F_i(x)$ places $x$ and $\bar x$ in the indicated off-diagonal positions and $L_A(X)=A\circ X$ is Jordan left multiplication.  Define
\begin{equation}
 \boxed{
 {\cal D}_{12}(z)=2[L_{p_1-p_2},L_{F_3(z)}]\in\mathfrak f_4(\mathbb C) .}
 \label{eq:D12Albert}
\end{equation}
Because a commutator of Jordan multiplication operators is an inner derivation, ${\cal D}_{12}(z)$ preserves the Albert product and cubic norm.  Direct multiplication gives
\begin{align}
 {\cal D}_{12}(z)F_1(x)&=F_2(\bar x\,\bar z),
 &
 {\cal D}_{12}(z)F_2(y)&=-F_1(\bar z\,\bar y).
 \label{eq:D12action}
\end{align}
For example, $F_3(z)\circ F_1(x)=\tfrac12F_2(\bar x\bar z)$, while $L_{p_1-p_2}$ has eigenvalues $-\tfrac12$ and $+\tfrac12$ on $F_1$ and $F_2$, respectively; Eq.~\eqref{eq:D12action} follows, including its normalization and signs.  This proof uses only products of two octonions and hence does not hide an associativity assumption.

Now take the exact Teli--Singh rung direction
\begin{equation}
 g_\chi=\cos\chi\,e_3-\sin\chi\,e_1,
 \qquad z_{12}=-g_\chi,
 \label{eq:gchiAlbert}
\end{equation}
and adopt the following physical-family embedding in the two fibres joined by Eq.~\eqref{eq:D12action}: $F_1(u_1),F_1(\bar d_1)$ represent the first family and $F_2(u_2),F_2(\bar d_2)$ the second \Ptag.  The Albert algebra does not by itself derive this assignment or make it unique.  The octonionic coefficients are
\begin{equation}
 u_1=e_4+ie_5,\quad u_2=e_6+ie_2,
 \qquad
 \bar d_1=e_5+ie_4,\quad \bar d_2=e_2+ie_6.
 \label{eq:localquarkstates}
\end{equation}
With the Fano convention of Ref.~\cite{TeliSinghCP}, the derivation acts exactly as
\begin{equation}
 {\cal D}_{12}(-g_\chi)F_1(u_1)=F_2(e^{-i\chi}u_2),
 \qquad
 {\cal D}_{12}(-g_\chi)F_1(\bar d_1)=F_2(e^{+i\chi}\bar d_2).
 \label{eq:Alberttransport}
\end{equation}
On each normalized two-state subspace its exponential is therefore
\begin{equation}
 e^{\theta{\cal D}_{12}}F_1(u_1)
 =\cos\theta\,F_1(u_1)+\sin\theta\,F_2(e^{-i\chi}u_2),
 \label{eq:Albertrotor}
\end{equation}
with the conjugate phase in the anti-down channel.  The local theorem
\begin{equation}
 \boxed{\varphi_{12}=-2\chi}
 \label{eq:Albertphaselaw}
\end{equation}
is thus promoted, for the adopted $F_1$--$F_2$ family embedding, from componentwise octonionic transport to a genuine Peirce-changing generator in the complex-linear extension of the Albert derivation algebra.  In particular, the operation on that chosen Peirce pair is no longer represented by an arbitrary tensor product $E_{12}\otimes L_{g_\chi}$.  What remains open is the vacuum theorem that identifies this pair with the physical first and second families.

This is a partial bridge theorem, not yet the full Yukawa kernel.  The order-three automorphism $\Gamma$ of Ref.~\cite{TeliSinghCP} gives canonical algebraic companions
\begin{equation}
 z_{23}=\Gamma z_{12}=\cos\chi\,e_6+\sin\chi\,e_1,
 \qquad
 z_{31}=\Gamma^2z_{12}=\cos\chi\,e_4+\sin\chi\,e_1,
 \label{eq:cyclicdirections}
\end{equation}
and corresponding derivations obtained by cyclically permuting $(p_i,F_i)$.  However, $\Gamma$ does not carry the displayed coloured quark representatives through three physical generations \cite{TeliSinghCP,TeliSinghMixing}.  Equations~\eqref{eq:cyclicdirections} therefore fix canonical Albert support, not the physical $23$ and $31$ quark matrix elements.  Nor do they determine the soldering $\Sigma:J_R\to J_L$.  Once such a locking map is supplied, skew-adjointness of the derivations fixes the reverse bridge by Hermitian conjugation, but it does not select $\Sigma$ itself.

\subsection{Masses and CKM mixing through the same finite Dirac operator}

The synthesis is exact as an operator identity, although it is a kinematic reorganization rather than a new derivation of CKM. Define
\begin{equation}
H_f=Y_fY_f^\dagger
=y_{f3}^2U_f\mathcal Q_f^2U_f^\dagger,
\qquad
P^L_{fi}=U_f|i\rangle\langle i|U_f^\dagger.
\label{eq:Hfprojectors}
\end{equation}
The eigenvalues of $H_f$ are the squared Yukawa singular values and hence encode the mass ratios; the projectors $P^L_{fi}$ encode the left-handed mass frame. Choosing normalized eigenvectors gives the CKM overlap matrix
\begin{equation}
(V_{\rm CKM})_{ij}=\langle u_i^L|d_j^L\rangle,
\qquad V_{\rm CKM}=U_u^\dagger U_d.
\label{eq:DFCKM}
\end{equation}
understood up to the usual independent diagonal rephasings of the up and down eigenvectors. The physical content is already expressible without choosing those phases:
\begin{equation}
|V_{ij}|^2=\operatorname{Tr}(P^L_{ui}P^L_{dj}),
\qquad
V_{ij}V_{kj}^*V_{kl}V_{il}^*
=\operatorname{Tr}(P^L_{ui}P^L_{dj}P^L_{uk}P^L_{dl}).
\label{eq:quartetprojectors}
\end{equation}
The imaginary parts of the second expression contain the rephasing-invariant $CP$ information. In this precise sense the mass-ratio and CKM constructions are the spectral and relative-frame parts of one finite Dirac problem. The statement is organizational, but nontrivial: it identifies the missing weak-to-mass bridge as the dynamical law that must select spectral projectors, rather than as another correction to the already-fixed eigenvalues. It does not imply that $\mathcal Q_f$ determines $U_f$.

\subsection{Relation to almost-commutative noncommutative geometry and novelty}
\label{sec:ncg}

The placement of masses and mixing in one finite Dirac operator has an important precedent. In the almost-commutative construction of Chamseddine, Connes and Marcolli, the admissible finite Dirac operators form a moduli space whose physical coordinates include quark and lepton masses, CKM and PMNS angles, and Majorana data \cite{CCM}. Thus neither the statement that masses and mixing reside in $D_F$, nor their separation into singular values and relative flavour frames, is claimed here as new. That construction also gives genuine aggregate restrictions, including gauge- and Higgs-sector boundary conditions and a unification-scale sum rule for squared Yukawa masses. It would therefore be incorrect to say that conventional noncommutative geometry makes no particle-physics predictions.

The distinction concerns the individual flavour data. In the standard model, the complex matrices $Y_u$, $Y_d$ and $Y_e$ are independent renormalizable couplings. Choosing them appropriately reproduces the observed masses and CKM matrix, but the standard-model field content and symmetries do not select their hierarchical singular values or relative left frame. Recasting the same data as coordinates of a finite geometry gives them a geometric home but, by itself, does not calculate their observed numerical values. Later Hodge-duality work shows that distinct nonzero masses and nontrivial mixing can be required for a refined property of the finite geometry, but does not derive the measured hierarchies or CKM entries \cite{DabrowskiSitarz}.

The new proposal is a partial selection principle on this finite-geometric moduli space:
\begin{equation}
X_f\ \longrightarrow\
P_{\rm occ}^{f}\Gamma_3(\widehat X_f)P_{\rm occ}^{f}\ \longrightarrow\
\mathcal Q_f\ \longrightarrow\
Y_f=y_{f3}U_f\mathcal Q_fW_f^\dagger.
\label{eq:ncgselection}
\end{equation}
Conditional on the parent Jordan spectra, chain assignments and occupied-node projectors, $\mathcal Q_f$ fixes two charged-fermion mass ratios per sector and the cross-sector relations of Refs.~\cite{Singh2508,TeliSinghMass}. Conditional also on the transport, virtual-node amplitude and quadrature postulates, the relative left frames give the Layer-1 values of $\Vus$ and $\Vcb$. The virtual node supplies a proposed structural origin for the up--down asymmetry of the $(2,3)$ mixing, while Lemma~\ref{lem:nogo} shows that a phase on that edge cannot repair $\Vub$ without spoiling $\Vcb$; the missing information is thereby localized in a complex long-edge bridge.

To our knowledge, this is the first exceptional-Jordan construction in which proposed charged-fermion mass-ratio moduli and partially constrained CKM frame data are organized as complementary spectral data of a finite Dirac operator. The priority claim is deliberately this narrow. Equation~\eqref{eq:ncgselection} does not yet give the full $D_F$: the chain and amplitude prescriptions remain programme inputs, $W_f$ is unknown, and the long-edge parameters are fitted. What it adds to conventional finite geometry is therefore a conditional calculation of part of the individual flavour data, not a completed exceptional spectral triple.

\begin{proposition}[Finite-Dirac underdetermination]
\label{prop:underdetermination}
The Jordan spectrum and occupied chains determine $\mathcal Q_f$, hence the two independent mass ratios in each charged sector. They do not determine $W_f$, and they do not derive the postulated left transport $U_f$. Consequently they select a family of finite Dirac operators, not a unique $D_F$. In particular, mass spectra alone cannot imply CKM mixing, texture zeros, a Yukawa-loop phase, or $\arg\det(Y_uY_d)$. \Dtag
\end{proposition}

\begin{proof}
For every $W_f\in U(3)$, Eq.~\eqref{eq:Ysvd} gives
$Y_fY_f^\dagger=y_{f3}^2U_f\mathcal Q_f^2U_f^\dagger$, so the masses and left diagonalizer are unchanged, while the entries, right-handed orientation and determinant phase of $Y_f$ vary. Moreover, $\mathcal Q_f$ itself contains no information selecting $U_f$. The CKM matrix $U_u^\dagger U_d$ therefore requires the transport/bridge data independently of the mass modulus.
\end{proof}

One useful representative illustrates the force of this result. The parity-motivated Hermitian completion $W_f=U_f$ gives $Y_f=y_{f3}U_f\mathcal Q_fU_f^\dagger$. At the fitted Layer-2 point, all nine entries of each of $Y_u$ and $Y_d$ are nonzero: the pair has eighteen entries, not the nine of the minimal textures in \cite{AFHM}. For the candidate invariant used in Sec.~\ref{sec:afhm},
\begin{equation}
\Phi_Y=\arg\!\left(Y^u_{12}Y^{u*}_{22}Y^d_{22}Y^{d*}_{12}\right),
\label{eq:PhiY}
\end{equation}
this representative gives exactly $90^\circ$ before the direct long edge is inserted, but $100.03^\circ$ on the orientation-A representative $(\varepsilon,\omega)=(0.002079,288.2^\circ)$. Its determinants are positive, so $\arg\det(Y_uY_d)=0$ at tree level. None of these three properties is a prediction of the mass modulus: other $W_f$ change them, and the opposite balanced orientation gives a different fitted bridge. The calculation is instead a counterexample to any claim that the present spectrum and transport automatically produce a sparse nine-link ultraviolet texture. A first-principles weak-to-mass bridge must fix both chiral frames, or an equivalent full Yukawa kernel, before that correspondence can be tested.

\section{Two extensions of the lift, and the one that survives}
\label{sec:transport}

\subsection{Adjacency from minimality}

The minimality principle that selects the chains \cite{Singh2508} also constrains the mixing structure: every physical step is a move along one edge $SU(2)$ of $\Sym^3(\mathbf 3)$, so whatever operator connects generations acts \emph{only between adjacent chain nodes}. Two readings of ``the mass matrix is built from adjacent steps'' are then available:

\emph{(i) Hamiltonian-texture reading}: the mass matrix in the flavour basis is a tridiagonal (nearest-neighbour) texture on the chain, with eigenvalues matched to the physical masses; diagonalize it. \emph{(ii) Transport reading}: the flavour-to-mass misalignment is the ordered product of two-state rotations, one per rung, each given exactly by the lift \eqref{eq:lift}; mass eigenstates are the chain nodes themselves, as the ladder normalization dictates (each generation's square-root mass \emph{is} its node monomial).

\subsection{The Hamiltonian-texture extension fails}

Reading (i) can be implemented in closed form: for the up chain the $3\times3$ zero-diagonal-plus-endpoint texture reproduces the classic Fritzsch results, but for the down chain the eigenvalue-matched $4\times4$ textures (the fourth eigenvalue assigned to the virtual node, all sign patterns and several auxiliary-eigenvalue candidates scanned) produce non-unitarity of $20$--$69\%$ across the ten admissible sign assignments, with $\Vus$ ranging over $0.05$--$0.66$ and the assignment closest to the Cabibbo value carrying $31\%$ non-unitarity; no viable case exists among the tested sign patterns and auxiliary-eigenvalue choices (Appendix~\ref{app:failed}) \Dtag. The failure is structural: matching four hierarchical eigenvalues with zero diagonals forces large hops and $\mathcal O(1)$ eigenvector rotations. We record this prominently, with the caveat that it excludes this specific (Fritzsch-rigid) texture family, not reading (i) in general: nearest-neighbour textures with free diagonals can always be made diagonal-dominant and mild. The adoption of reading (ii) is thus a modelling choice supported by the exclusion and by the transport character of the lift itself \cite{Fritzsch,GST}, not a forced conclusion.

\subsection{The transport model}
\label{sec:transportmodel}

\begin{postulate}[Transport extension of the adjacent-edge lift]
\label{post:transport}
For each sector, the unitary relating the flavour basis to the mass basis is the ordered product of rung rotations along the minimal chain, applied anchor-outward (the rung adjacent to the first generation acts first),
\begin{equation}
U_f\;=\;R_{23}\big(\theta^{f}_{23}\big)\,R_{12}\big(\theta^{f}_{12}\big),
\qquad \tan\theta^{f}_{ij}\ \text{from \eqref{eq:rungs}},
\label{eq:Uf}
\end{equation}
where for a chain segment that traverses unoccupied nodes the rotation angle is defined by the composed rung rotations under the amplitude reading (Postulate~\ref{post:comp} below), and $V_{\rm CKM}=U_u^\dagger U_d$.
\end{postulate}

Three consequences are immediate. The $(1,2)$ blocks reproduce the exact two-state lift (GST relations) by construction. The $(1,3)$ rotations are absent: the long edge of each sector is \emph{automatically} the composed product of its adjacent transports, which appears in \eqref{eq:Uf} as the $(3,1)$ matrix element $s_{12}s_{23}$; no separate long-edge input exists at this layer. And the model is exactly unitary, with all angles in closed form.

\begin{postulate}[Virtual-node composition: the amplitude reading]
\label{post:comp}
Let the $2\to3$ segment of a chain traverse one unoccupied node, with rung angles $\theta_2$ (occupied$\,\to\,$virtual) and $\theta_3$ (virtual$\,\to\,$occupied). The effective physical $2$--$3$ mixing angle is \emph{defined} by the two-step projected amplitude of the composed transport,
\begin{equation}
\boxed{\ \sin\theta_{23}^{\rm eff}\;=\;\big\langle\,3\,\big|\,R_{v3}(\theta_3)R_{2v}(\theta_2)\,\big|\,2\,\big\rangle\;=\;\sin\theta_2\,\sin\theta_3\ .}
\label{eq:comp}
\end{equation}
\end{postulate}

\noindent We are explicit about the status of \eqref{eq:comp}, because an earlier draft of this work claimed it as a theorem and the claim does not hold. The composed transport, projected onto the physical (occupied-node) subspace, is a \emph{triangular} $2\times2$ block, not a rotation: one off-diagonal is $\sin\theta_2\sin\theta_3$, the other vanishes, and the block is not unitary (the virtual admixture removed by the projection carries $26\%$ of the second state's norm). Different re-unitarization conventions therefore assign different angles: the polar decomposition (nearest unitary) gives $\sin\theta_{\rm eff}=0.0675$, row normalization $0.1257$, column normalization $0.1435$, against the amplitude value $0.1232$; propagated to the observable these give $\Vcb=0.014$, $0.045$, $0.063$ and $0.042$ respectively. The amplitude reading \eqref{eq:comp} is selected here as a kinematic composition rule because it directly retains the two-step matrix element and recovers the adjacent-edge two-state limit smoothly.  It is not identified with a second-order Hamiltonian amplitude, for which energy denominators and microscopic dynamics would be required.  The selection is a postulate \Ptag, and a first-principles bridge computation must ultimately decide it. Where we write ``derived'' below, it means: derived given \eqref{eq:comp}.

\section{The phase sector: one rule, one no-go}
\label{sec:phases}

A word first on where phases can and cannot live, for the reader outside the texture literature. Rephasing the quark fields acts like a gauge transformation at the nodes of the flavour graph: any phase sitting on an open path can be rotated away, and only the phase around a closed loop is physical \cite{Jarlskog,BLM}. In the transport model the up and down chains share their first rung, and that shared edge is exactly what closes a loop between the two sectors. The postulate below says that all the physical phase content of the model lives on that loop, and the theorem of \cite{TeliSinghCP} fixes its value.

\begin{postulate}[Shared-edge phases; balanced rotor]
\label{post:phase}
Relative phases between the up and down transports enter only on the \emph{shared} Cabibbo edge (both chains' first rung is the same $E$-edge $a^2b\to abc$), where the conjugation theorem of \cite{TeliSinghCP} gives $A_d=A_u^*$ and hence $\varphi_{12}=-2\chi$; the balanced rotor fixes $|\varphi_{12}|=90^\circ$. The $(2,3)$ rungs of the two sectors lie on \emph{distinct} edges ($C$ versus $B$), the conjugation theorem does not apply, and the corresponding blocks are relatively real.
\end{postulate}

The first clause is the Letter's derived law plus the balanced-rotor choice \Ptag. The second clause, minimal on its face, is in fact \emph{forced} by data within this model:

\begin{lemma}[No-go for a relative $(2,3)$ phase]
\label{lem:nogo}
In the transport model, a relative phase $\psi$ between the down ($C$-edge) and up ($B$-edge) $(2,3)$ blocks interpolates
\begin{center}
\begin{tabular}{lcccc}
$\psi$ & $0^\circ$ & $30^\circ$ & $45^\circ$ & $90^\circ$\\
$\Vcb$ & $0.0422$ & $0.0665$ & $0.0870$ & $0.1467$\\
$\Vub$ & $0.0018$ & $0.0028$ & $0.0037$ & $0.0062$
\end{tabular}
\end{center}
and the pattern in the table is exact: with phases confined to the $(1,2)$ blocks and a relative $(2,3)$ phase $\psi$, the long-edge elements share the common factor $D(\psi)=s^d_{23}c^u_{23}-e^{i\psi}c^d_{23}s^u_{23}$, namely $V_{ub}=-s^u_{12}e^{-i\varphi}D(\psi)$ and $V_{cb}=c^u_{12}D(\psi)$, so the ratio $|V_{ub}/V_{cb}|=\tan\theta^u_{12}$ holds for \emph{every} $\psi$, while $|D(\psi)|$ is strictly increasing on $[0,\pi]$. Hence no $\psi$ can repair $\Vub$ without inflating $\Vcb$ by the same factor: the $(2,3)$ interference must remain destructive and real, and the $\Vub$ deficit must be repaired elsewhere. \Dtag
\end{lemma}

Lemma~\ref{lem:nogo} is the sharpened form of the companion paper's ``long-edge bridge problem'': the missing structure provably cannot sit in the $(2,3)$ phase sector.

\section{\texorpdfstring{$\kappa_{23}$}{kappa23} from the virtual node}
\label{sec:kappa}

Applying Postulate~\ref{post:comp} to the down chain with the rung tangents \eqref{eq:rungs}:
\begin{equation}
\sin\theta^{d}_{23}=\sin\theta_2\sin\theta_3=0.1232,
\qquad
\sqrt{m_s/m_b}=\tan\theta_2\tan\theta_3=0.1491,
\end{equation}
(the naive value is the telescoping product of tangents; the composition replaces tangents by sines, a sine-to-sine ratio of $0.835$). The up chain has no virtual node: $\sin\theta^{u}_{23}=\sin(\arctan 0.0814)=0.0812$, unsuppressed. In the difference form of \cite{TeliSinghMixing} (their Eq.~(49)),
\begin{equation}
\Vcb\;=\;\big|\,s^{d}_{23}c^{u}_{23}-c^{d}_{23}s^{u}_{23}\,\big|\;=\;0.0422,
\qquad
\kappa_{23}^{\rm eff}\equiv\frac{\Vcb}{\big[s^{d}_{23}c^{u}_{23}-c^{d}_{23}s^{u}_{23}\big]_{\rm naive}}=0.633 .
\label{eq:vcb}
\end{equation}
Compared with the 2026 global-fit $\Vcb=0.04188$: $+0.8\%$, with no fitted continuous CKM quantity after the stated structural choices \Dtag~(given Postulate~\ref{post:comp}). Remarks. (1)~The companion paper's fitted $\kappa_{23}\simeq0.56$ was defined with the \emph{fitted} mass ratios of \cite{TeliSinghMass}; with the theoretical ratios used throughout here, the value required by data is $0.62$--$0.63$, which the amplitude composition delivers. The composition does \emph{not} retrodict $0.56$ in the companion's own fitted frame (it gives $0.685$ there); the two statements live in different input frames, and only the all-theory frame has no fitted continuous CKM parameter after adopting the stated transport and amplitude postulates. (2)~The deliverable was pass/fail: the verdict is a \emph{conditional pass}: pass under the amplitude reading, with the alternative readings and their spread disclosed in Sec.~\ref{sec:transportmodel}. (3)~The same mechanism predicts the \emph{absence} of suppression in the up sector, and, by the Dynkin reflection, a suppression of the same type on the charged-lepton chain (Sec.~\ref{sec:ledger}).

\paragraph{Fragility disclosure.}
The difference form \eqref{eq:vcb} is a small difference of larger terms, with elasticity $\partial\ln\Vcb/\partial\ln t^{u}_{23}=-1.9$: the result is hostage to the parent theory's $\sqrt{m_c/m_t}$. That input is itself the parent's worst entry (theory $1/12.28$ against the measured $1/16.65$, a $26\%$ deviation \cite{TeliSinghMass}), and the three input frames give
\begin{center}
\begin{tabular}{lccc}
\toprule
frame for $\sqrt{m_c/m_t}$ & theory ($0.0814$) & companion fit ($0.0722$) & experiment ($0.0601$)\\
$\Vcb$ (amplitude reading) & $0.0422$ & $0.0515$ & $0.0635$\\
\bottomrule
\end{tabular}
\end{center}
The $+0.8\%$ agreement therefore holds in, and only in, the internally consistent all-theory frame, where the parent's $m_t/m_c$ overshoot and the virtual-node suppression partially compensate. We state this plainly rather than let the headline number stand unqualified: the correlation between the $\Vcb$ prediction and the parent's $\sqrt{m_t/m_c}$ tension is itself a testable internal relation (Sec.~\ref{sec:ledger}), exactly parallel to the $\Vus$--$\sqrt{m_s/m_d}$ correlation noted below.

\section{The conditional no-fit layer}
\label{sec:assembly}

Assembling $V=U_u^\dagger U_d$ with Postulates~\ref{post:transport}--\ref{post:phase} and nothing else gives the following conditional no-fit layer.  Here ``no-fit'' means no fitted \emph{continuous} CKM parameter after adopting the parent Jordan spread and node assignments, the occupied-node transport reading, the virtual-node amplitude rule, relatively real $(2,3)$ blocks, and one of the two balanced shared-edge orientations.  These are substantive structural choices, not five interchangeable binary parameters.

\begin{table}[t]
\centering
\small
\begin{tabular}{lccc>{\raggedright\arraybackslash}p{0.20\textwidth}}
\toprule
observable & model & PDG global fit & deviation & status \\
\midrule
$\Vus$ & $0.2371$ & $0.22517$ & $+5.3\%$ & independent angle estimate \\
$\Vcb$ & $0.0422$ & $0.04188$ & $+0.8\%$ & independent angle estimate \\
\midrule
$\Vts$ & $0.0411$ & $0.04115$ & $-0.2\%$ & unitarity consequence \\
$\Vtd$ & $0.0099$ & $0.00863$ & $+14.5\%$ & unitarity consequence \\
$A$ (Wolfenstein) & $0.750$ & $0.826$ & $-9.1\%$ & $\equiv\Vcb/\Vus^2$ \\
$\Vub$ & $0.0018$ & $0.003762$ & $-52.4\%$ & Fritzsch-ratio consequence \\
$J$ & $1.72\times10^{-5}$ & $3.16\times10^{-5}$ & $-45.6\%$ & unitarity/phase consequence \\
\bottomrule
\end{tabular}
\caption{Layer 1: two conditional angle estimates with no fitted continuous CKM parameters after the listed structural choices.  The remaining rows are consequences of the adopted unitary construction; $A$ is definitional, not an independent prediction.  Everything involving the long edge inherits the Fritzsch deficit.}
\label{tab:layer1}
\end{table}

\begin{lemma}[Fritzsch ratio]
\label{lem:fritzsch}
In the minimal transport model the long-edge elements obey, exactly,
\begin{equation}
\frac{\Vub}{\Vcb}\;=\;\tan\theta^{u}_{12}\;=\;\sqrt{\tfrac{m_u}{m_c}}\;=\;0.04245 ,
\end{equation}
against the 2026 global-fit ratio $0.090$: the classic square-root-texture shortfall, here not an accident of an ansatz but an exact identity of the minimal model ($V_{ub}=-s^{u}_{12}e^{-i\varphi}D$ and $V_{cb}=c^{u}_{12}D$ for a common factor $D$, so the ratio is exactly the tangent). \Dtag
\end{lemma}

The $\Vus$ value deserves its own comment. $0.2371$ versus $0.22517$ is a $+5.3\%$ deviation, far outside experimental error, and comparable to the mass-ratio tensions of the parent construction (e.g.\ $\sqrt{m_s/m_d}$ at $-4.5\%$ against the companion's $M_Z$ baseline \cite{TeliSinghMass}). Since $\Vus$ here is built \emph{from} those mass ratios, the two deviations are correlated, and the proto-$\sqrt m$ matching systematic flagged in \cite{Singh2508,TeliSinghMass} feeds directly into the Cabibbo angle. We therefore quote $+5.3\%$ as a real, owned tension of the conditional no-fit layer, not as a success. Matching $\Vus$ exactly with the theoretical ratios would require $\varphi_{12}\simeq107^\circ$ (the theory-frame analogue of the companion's fitted $105.7^\circ$ \cite{TeliSinghMixing}); the balanced value $90^\circ$ is retained here because it is the distinguished no-fit point, and $\Vus$ is, notably, the \emph{only} observable in the construction that is sensitive to this choice (Sec.~\ref{sec:bridge}).

\section{The long edge: one complex number, eight observables}
\label{sec:bridge}

Lemma~\ref{lem:nogo} confines the missing structure to the $(1,3)$ sector, and the transport model says what is missing: a \emph{direct} long-edge rotation, over and above the composed adjacent transport; this is precisely the object the companion paper called the long-edge bridge problem. We therefore add the minimal such term, an up-sector rotation $R_{13}(\varepsilon,\omega)$ inserted between the rungs,
\begin{equation}
U_u\;=\;R_{23}\big(\theta^{u}_{23}\big)\,R_{13}(\varepsilon,\omega)\,R_{12}\big(\theta^{u}_{12}\big).
\end{equation}
The parameter counting must be stated before the numbers: a unitary $3\times3$ mixing matrix has exactly four physical parameters, of which Layer 1 conditionally fixes two ($\theta_{12}$, $\theta_{23}$) without fitted continuous CKM parameters after the structural choices listed above.  The insertion supplies the remaining two, $\varepsilon\leftrightarrow\theta_{13}$ and $\omega\leftrightarrow\delta_{CP}$, and both are \emph{fitted}.  Table~\ref{tab:layer2} displays orientation A, obtained by solving exactly on $\Vub$ and $\delta_{CP}$: $\varepsilon=0.00207936$ and $\omega=288.1915^\circ$. Pairing $\omega$ with any other CP-sector quantity gives the same physics; within a fixed convention, $\omega$ is the dial for the CKM phase. Every remaining CKM quantity is then a function of the four parameters so fixed:

\begin{table}[t]
\centering
\small
\begin{tabular}{lcccl}
\toprule
observable & model & PDG global fit & deviation & status \\
\midrule
$\Vus$ & $0.2371$ & $0.22517$ & $+5.3\%$ & Layer 1 ($\theta_{12}$) \\
$\Vcb$ & $0.0421$ & $0.04188$ & $+0.6\%$ & Layer 1 ($\theta_{23}$) \\
$\Vub$ & $0.003762$ & $0.003762$ & --- & fit target ($\varepsilon$) \\
$\delta_{CP}$ & $66.1^\circ$ & $66.1^\circ$ & --- & fit target ($\omega$) \\
\midrule
$\Vtd$ & $0.00914$ & $0.00863$ & $+6.0\%$ & unitarity consequence \\
$\Vts$ & $0.0413$ & $0.04115$ & $+0.3\%$ & unitarity consequence \\
$J$ & $3.34\times10^{-5}$ & $3.16\times10^{-5}$ & $+5.6\%$ & unitarity consequence \\
$\beta$ & $21.5^\circ$ & $22.9^\circ$ & $-6.2\%$ & unitarity consequence \\
$\gamma$ & $66.1^\circ$ & $66.1^\circ$ & $+0.0\%$ & $\simeq\delta_{CP}$ (to $0.03^\circ$) \\
$\alpha$ & $92.4^\circ$ & $91.0^\circ$ fit / ${\sim}85^\circ$ dir. & $+1.6\%$/$+9\%$ & $180^\circ{-}\beta{-}\gamma$ \\
$\bar\rho$ & $0.148$ & $0.1576$ & $-5.9\%$ & function of $(R_b,\gamma)$ \\
$\bar\eta$ & $0.335$ & $0.3556$ & $-5.9\%$ & function of $(R_b,\gamma)$ \\
$A$ & $0.749$ & $0.826$ & $-9.3\%$ & $\equiv\Vcb/\Vus^2$ \\
\bottomrule
\end{tabular}
\caption{The four CKM degrees of freedom (upper block: two conditional no-fit estimates, two fitted quantities) and the derived quantities that follow from them by unitarity (lower block). The lower-block residuals are the propagation of the Layer-1 moduli deviations: consistency checks, \emph{not} independent predictions; e.g.\ the $A$ deviation is exactly $(1+\delta_{\Vcb})/(1+\delta_{\Vus})^2-1$.}
\label{tab:layer2}
\end{table}

We do not claim the lower block of Table~\ref{tab:layer2} as evidence: once the four degrees of freedom are fixed, unitarity guarantees that every derived quantity agrees with a unitarity-based global fit to the extent that the four inputs do. In particular, the CP sector \emph{cannot} discriminate the balanced quadrature phase from nearby alternatives: refitting $(\varepsilon,\omega)$ at $\varphi_{12}$ anywhere in a broad band reproduces the CP set equally well, because $\omega$ dials $\delta_{CP}$ over its full range. The one observable in the entire construction that is sensitive to $\varphi_{12}$ is $\Vus$ itself, and it prefers $\varphi_{12}\simeq107^\circ$ over the balanced $90^\circ$ at the $+5.3\%$ level noted above. The honest summary of Layer 2 is therefore modest: the two remaining degrees of freedom are fitted, both belong to one complex matrix element of the weak-to-mass bridge, and the derived block confirms internal coherence at the few-percent level set by Layer 1.

There is also a discrete ambiguity that a single-orientation solve hides.  Reversing the balanced local rotor and repeating the full-circle fit gives a second physical solution with the same positive target $\delta_{CP}$:
\begin{center}
\begin{tabular}{lcccc}
\toprule
balanced orientation & $\varepsilon$ & $\omega$ & $J$ & $\varepsilon/(s^u_{12}s^u_{23})$ \\
\midrule
A (used in Table~\ref{tab:layer2}) & $0.00207936$ & $288.1915^\circ$ & $3.335\times10^{-5}$ & $0.604$ \\
B (opposite orientation) & $0.00535850$ & $202.1844^\circ$ & $3.358\times10^{-5}$ & $1.557$ \\
\bottomrule
\end{tabular}
\end{center}
Both branches reproduce $\Vub=0.0037623$ and $\delta_{CP}=66.13^\circ$ to numerical precision.  The local algebra distinguishes the two orientations, but the present vacuum dynamics does not select one.  Orientation A is retained as the displayed representative because it continues the convention used in the preceding analysis, not because the data or the bridge dynamics uniquely prefer it.

The branch comparison removes two apparent numerical clues.  On orientation A,
\begin{equation}
\frac{\varepsilon}{s^{u}_{12}s^{u}_{23}}\;=\;0.604 ,
\end{equation}
but orientation B gives $1.557$.  Thus the proximity of $0.604$ to $3/5$ or to other programme constants is branch-dependent numerology, not a robust target.  Moreover, the numerical value of the internal phase $\omega$ depends on the rotation ordering, quark rephasings and the convention used to distribute the Cabibbo phase.  Only rephasing-invariant CKM observables are physical; no Fano-geometric meaning is assigned here to the raw $\omega$ offset.

\section{The phase sector in the light of the nine-link texture scan}
\label{sec:afhm}

While this work was between its Zenodo release and the present arXiv posting,\footnote{The construction and its conditional no-fit layer are unchanged from the Zenodo release of 4 July 2026 (doi:\href{https://doi.org/10.5281/zenodo.21196380}{10.5281/zenodo.21196380}), which predates \cite{AFHM}. Relative to that release, the present version adds the finite-Dirac operator audit of Sec.~\ref{sec:finite} and the resulting explicit test of the proposed nine-link dictionary, corrects a sign error in the Appendix~\ref{app:failed} scan (the exclusion stands, with revised ranges), replaces the coarse-grid Layer-2 fit by exact solves on $(\Vub,\delta_{CP})$ against an internally coherent PDG 2026 target set, discloses the two balanced-orientation branches, and upgrades the no-go lemma to its exact form.} Arkani-Hamed, Figueiredo, Hall and Manzari posted a bottom-up analysis \cite{AFHM} that bears directly on the phase sector of this paper. They parametrize the ten quark flavour observables by nine-link textures (full-rank $Y_{u,d}$ with nine nonzero entries in total and a single rephasing-invariant phase $\varphi$), scan the full space of such textures, and report that the fitted $\varphi$ does not distribute featurelessly but clusters tightly at multiples of $\pi/8$, tracking the empirical observation that the unitarity-triangle angles lie within current errors of $(\pi/2,\pi/8,3\pi/8)$ (two independent coincidences rather than three, since $\alpha+\beta+\gamma=\pi$); the histogram also shows a secondary peak at $\varphi\sim\pi/4$, which they attribute to a coincidence at the level of the quark masses, not directly connected to the unitarity triangle. Each texture class predicts its \emph{special} angle with quoted spreads as small as hundredths of a degree (the other two angles are less sharply fixed), decisively testable as Belle~II and LHCb shrink the errors by an order of magnitude; and for all but five of their ninety-nine fixed-phase textures the sparseness makes $\det(Y_uY_d)$ automatically real, so that with spontaneous $CP$ violation $\bar\theta$ vanishes at tree level with no axion. The observation that the unitarity triangle may be exactly right-angled has its own history \cite{Xing2009,Antusch2013,Mimura,HarrisonScott}, which the scan of \cite{AFHM} systematises.

Three points of contact, each with its caveat:

\begin{enumerate}
\item \textbf{Quantization found versus the operator audit.} Where their scan finds phase quantization empirically, the conjugation theorem of \cite{TeliSinghCP} permits one relative phase on the shared edge, and the balanced rotor sets its transport value to $|\varphi_{12}|=90^\circ=\pi/2$, one of the cluster values of \cite{AFHM}. The underlying $|\chi|=\pi/4$ also coincides numerically with their secondary $\pi/4$ peak, though the origins differ. Section~\ref{sec:albertbridge} now advances the dictionary beyond a spectral resemblance: ${\cal D}_{12}(-g_\chi)$ acts on the adopted Peirce pair and gives an explicit off-diagonal exceptional generator carrying the conjugate quark phases. What it still does not select is the physical family embedding or the chiral contraction that turns this generator into one particular matrix element of $Y_u$ or $Y_d$. The right frames $W_f$ remain free. In the explicit Hermitian representative $W_f=U_f$, their candidate loop $\arg(Y^{u}_{12}Y^{u*}_{22}Y^d_{22}Y^{d*}_{12})$ is indeed $90^\circ$ in the adjacent-only model, but the fitted direct long edge moves it to $100.03^\circ$; moreover, both Yukawa matrices are dense, with eighteen rather than nine nonzero entries in total. The correspondence is therefore not yet implied. A bridge kernel that selects the chiral projection, $W_f$, sparsity and the direct edge could place the result in one of their classes, but that remains a dynamical calculation \Otag. Texture zeros and loop phases are rephasing-invariant but not weak-basis-invariant; the Peirce frame is now a candidate ultraviolet basis, but it becomes physical only after the vacuum locks it to the chiral modules.
\item \textbf{Precision.} Their islands in $(\alpha,\beta,\gamma)$ are, in each class's special angle, hundredths-of-a-degree objects; the propagated angles of Table~\ref{tab:layer2} carry the few-percent residuals of the Layer-1 moduli and are consistency checks, not predictions (Sec.~\ref{sec:bridge}). The present construction does not currently play in that game, and we say so. What the comparison sharpens is the specification of the bridge computation: if the coming measurements lock the angles onto the special fractions with texture-calculable deviations, a first-principles $(\varepsilon,\omega)$ must reproduce exactly that structure, not merely the vicinity of the measured angles. A concrete instance of the future discrimination: their nine $\varphi=\pi/2$ classes predict $\alpha$ between $88.74^\circ$ and $90.03^\circ$, while Table~\ref{tab:layer2}'s $\alpha=92.4^\circ$ is a unitarity propagation of fitted and conditional no-fit moduli, not an independent prediction, so there is no present contradiction; a completed bridge computation would convert this into a sharp either/or: the reconstructed textures fall in a nine-link $\pi/2$ class or they do not.
\item \textbf{Strong $CP$.} Their sparsity mechanism ($\det(Y_uY_d)$ real by a single perfect matching) is independent of, and complementary to, the spacetime left--right-exchange protection of \cite{SinghLR}, in which one vector-like colour group is blind to the breaking and $\bar\theta=0$ at tree level, likewise with no axion. The Hermitian representative of Sec.~\ref{sec:finite} also has $\arg\det(Y_uY_d)=0$, but this follows from the added choice $W_f=U_f$, not from the Jordan modulus; general $W_f$ change the determinant phase. The two genuine mechanisms therefore retain the same open question: whether the completed bridge fixes the chiral frames in a protected class and whether that protection survives radiative corrections.
\end{enumerate}

The exact lift gives one additional graph-theoretic insight.  Before symmetry-breaking mixings are inserted, three positive singular values in each quark sector may be represented by six diagonal links joining $(Q_i,u_i^c)$ and $(Q_i,d_i^c)$.  These form three disconnected components on the nine chiral-family vertices.  If the physical finite Dirac block is \emph{linear} in three directed Peirce-changing contractions, one for each of the $12$, $23$ and $31$ Albert edges, those three additional links connect the components and give
\begin{equation}
 E=6+3=9,\qquad V=9,\qquad C=1,\qquad E-V+C=1.
 \label{eq:ninelinkcount}
\end{equation}
Thus a physical Yukawa block linear in the three projected cyclic Albert links would have the support count and unique-cycle topology of a nine-link texture. If the two non-Cabibbo contractions are real and the unique cycle contains the balanced Cabibbo link, Eq.~\eqref{eq:Albertphaselaw} would put its holonomy at $\pi/2$. This is a candidate ultraviolet route to the $\varphi=\pi/2$ classes of Ref.~\cite{AFHM}, and is stronger than the previous singular-value compatibility alone, but it remains conditional on the linear contraction and vacuum-selected chiral endpoints.

Two qualifications prevent promotion to a theorem.  First, a sparse generator is not a sparse finite rotation: $e^{\cal B}$ generically fills matrix entries.  Equation~\eqref{eq:ninelinkcount} applies only if the Yukawa block is linear in the projected bridge field (or obeys an equivalent nilpotent/selection rule), not if it is reconstructed as $e^{\cal B}{\cal Q}$.  Second, the cyclic companions in Eq.~\eqref{eq:cyclicdirections} have not been shown to be the physical coloured $23$ and $31$ contractions, and the right-sector locking remains open.  The lift therefore supplies a candidate nine-link \emph{skeleton}; the microscopic finite Dirac map must still decide whether that skeleton, rather than its dense exponential, is realized.

We performed an ancillary reduced-protocol central-value test of whether the exceptional-Jordan singular values obstruct nine-link Yukawa supports.  Exact enumeration gives 6552 labelled connected full-rank unicyclic supports and 36 orbits under $S_3^Q\times S_3^u\times S_3^d$.  At fixed loop phase we impose the four mass ratios and the three $CP$-even moduli $(\Vus,\Vcb,\Vub)$, with no low-energy $CP$ observable in the fit; the floating-phase scan additionally imposes the experimental $\gamma$, making its phase distribution conditional on the observed CKM triangle.  As a local convention check, the Fig.~2 texture of Ref.~\cite{AFHM} at $\varphi=\pi/2$ gives $\alpha=88.950^\circ$, against their $88.957^\circ$; one benchmark does not validate an exhaustive classification.  Replacing the measured $M_Z$ ratios by the Jordan ratios leaves regular exact solutions in the same 28 support orbits.  At the chosen central inputs, 53 matched deduplicated support--phase representatives have a mean absolute loop-phase displacement of $0.149^\circ$ (maximum $0.457^\circ$), and fixed-phase special-angle changes below $0.425^\circ$.  These are central-input sensitivities, not uncertainty-qualified physical precisions.  The test establishes compatibility, but it neither reproduces the full uncertainty-weighted classification of Ref.~\cite{AFHM} nor tests the exceptional-Jordan CKM transport, because experimental CKM data were imposed.  It does not derive nine-link sparsity, the preferred weak basis, right-handed frames or the loop phase.  The included raw outputs and self-contained summarizer reproduce this result.

The net effect of \cite{AFHM} is therefore threefold. It supplies external, bottom-up evidence that a Yukawa-loop phase may sit at a simple fraction of $\pi$, the kind of rigidity sought here; the Albert lift supplies a natural three-edge, one-cycle exceptional skeleton; and the finite-Dirac audit prevents us from identifying its local transport phase with a Yukawa-loop phase until the chiral projection and right-sector lock are derived.

\section{Ledger, lepton mirror, and future tests}
\label{sec:ledger}

\subsection{Assumptions ledger}

\begin{center}
\small
\begin{tabular}{>{\raggedright\arraybackslash}p{0.42\textwidth}>{\raggedright\arraybackslash}p{0.16\textwidth}>{\raggedright\arraybackslash}p{0.34\textwidth}}
\toprule
ingredient & status & source / target \\
\midrule
$\delta^2=3/8$; chains; monomial masses & \Ptag\ (parent) & \cite{Singh2508,TeliSinghMass} \\
two Albert copies and two spacetime leaves & \Ptag\ (parent structures) & \cite{Singh2508,Scaffolding,WesleySinghIsidroBF,SinghEmergence} \\
quadratic $J_2(\mathbb H_s)$ norm, cubic Albert norm and corner factorization & algebraic identities \Dtag; physical descent \Ptag & Sec.~\ref{sec:quadraticcubic} \\
$J_2(\mathbb O_s)=J_2(\mathbb H_s)\oplus\mathbb H_s\ell$ and the $B\psi\chi$ channel & representation identity \Dtag; physical localization \Otag & Sec.~\ref{sec:universalsoldering} \\
factorized bridge $Y_f=B\otimes\Xi_f$ & conditional ansatz \Ptag; cancellation \Dtag & Eq.~\eqref{eq:CKMXifactor}; $\Xi_f$ remains \Otag \\
branch pairing $(\Sigma_L,J_L)\leftrightarrow$ flavour, $(\Sigma_R,J_R)\leftrightarrow$ mass & \Ptag; gluing kernel \Otag & Sec.~\ref{sec:leavesDF} \\
$\mathcal R_f$ and $\mathcal Q_f=\mathcal R_f^\dagger\mathcal R_f/\rho_{f3}^2$ & \Dtag\ given parent spectrum & Prop.~\ref{prop:rootoperator} \\
finite $D_F$ family; right frames $W_f$ & modulus derived, frames \Otag & Prop.~\ref{prop:underdetermination} \\
Cabibbo Albert lift ${\cal D}_{12}(-g_\chi)$ & exact generator \Dtag; family embedding \Ptag & Sec.~\ref{sec:albertbridge} \\
cyclic $23,31$ Albert companions & support derived; coloured soldering \Otag & Eqs.~\eqref{eq:cyclicdirections} and \eqref{eq:MajoranaProjection} \\
minimal bifundamental vacuum & alignment or flatness \Dtag & App.~\ref{app:vacuum}; conditional on $U(3)^3$ field content \\
three-edge Albert potential & truncation result; sourced alignment in $W$ \Dtag & App.~\ref{app:edgepotential}; $W$ and source components \Otag \\
adjacent-edge lift $\tan\theta=\sqrt{m_i/m_j}$ & \Ptag, partially \Dtag & two-state texture exact; adjacency from minimality \\
transport reading (vs Hamiltonian texture) & modelling choice, supported by one exclusion & App.~\ref{app:failed} \\
amplitude reading of the virtual-node composition & \Ptag & Postulate~\ref{post:comp}; alternatives disclosed \\
$\kappa_{23}$ and $\Vcb$ $+0.8\%$ & \Dtag~(given Postulate~\ref{post:comp}) & Sec.~\ref{sec:kappa}, incl.\ fragility table \\
shared-edge phase law $\varphi_{12}=-2\chi$ & \Dtag\ (Letter) & \cite{TeliSinghCP} \\
balanced rotor orientation; sign of Layer-1 $J$ & \Ptag\ + discrete choice & both orientations disclosed in Sec.~\ref{sec:bridge} \\
real $(2,3)$ blocks & excluded-by-data alternative & Lemma~\ref{lem:nogo} (uses PDG $\Vcb$, $\Vub$) \\
long-edge bridge $(\varepsilon,\omega)$ = $(\theta_{13},\delta_{CP})$ & fitted (2 continuous numbers) & two discrete orientation branches; raw $\omega$ is convention-dependent \\
\bottomrule
\end{tabular}
\end{center}

\subsection{The lepton mirror}

The charged-lepton chain is the Dynkin reflection of the down chain \cite{Singh2508} and shares its four-node structure (virtual $ab^2$). The transport machinery therefore applies verbatim: the charged-lepton $(1,2)$ misalignment is $\arctan\sqrt{m_e/m_\mu}\simeq4.1^\circ$, and the $(2,3)$ misalignment carries the same virtual-node suppression as the down sector. In the PMNS matrix these are the small charged-sector corrections; the large angles remain condensate-frame quantities on the neutrino side, undetermined at this layer, exactly as established in the neutrino-sector paper \cite{SinghNu}.

The Albert lift adds a structural statement about that neutral frame.  Adopt the cyclic Peirce placement
\begin{equation}
 \nu_1=F_3(ie_7),\qquad
 \nu_2=F_1(ie_5),\qquad
 \nu_3=F_2(ie_2),
 \label{eq:MajoranaPeirce}
\end{equation}
and form one cyclic bridge generator
\begin{equation}
 {\cal B}=\rho_{12}{\cal D}_{12}(z_{12})
 +\rho_{23}{\cal D}_{23}(z_{23})
 +\rho_{31}{\cal D}_{31}(z_{31}).
 \label{eq:cyclicbridge}
\end{equation}
Direct projection onto the real span of Eq.~\eqref{eq:MajoranaPeirce} gives, up to simultaneous edge-orientation conventions,
\begin{equation}
 \boxed{
 P_\nu{\cal B}P_\nu
 =\cos\chi
 \begin{pmatrix}
 0&-\rho_{31}&\rho_{23}\\
 \rho_{31}&0&-\rho_{12}\\
 -\rho_{23}&\rho_{12}&0
 \end{pmatrix}.}
 \label{eq:MajoranaProjection}
\end{equation}
In the convention of Ref.~\cite{SinghNu}, this is precisely the real $(e_4,e_3,e_6)$ support,
\begin{equation}
 b_4=-\rho_{31}\cos\chi,\qquad
 b_3=\rho_{12}\cos\chi,\qquad
 b_6=-\rho_{23}\cos\chi.
 \label{eq:MajoranaSupport}
\end{equation}
The $e_1$ quadrature that carries the local quark phase projects to zero.  Therefore the minimal common bridge is compatible with a real PMNS representative and
\begin{equation}
 J_\ell=0,\qquad \delta^\ell_{CP}\in\{0,\pi\},
 \label{eq:leptonCPsafe}
\end{equation}
without reintroducing the withdrawn maximal-leptonic-$CP$ claim.  This conclusion is conditional on the Majorana placement and on the bridge remaining in the projected real-linear class; an intrinsically complex bridge or identity--flavour mixing evades it \cite{TeliSinghCP,SinghNu}.

Equation~\eqref{eq:MajoranaProjection} fixes support but not magnitudes.  The three real numbers $\rho_{12},\rho_{23},\rho_{31}$ are independent Peirce-edge strengths until a vacuum equation determines them.  In particular, neither the corner amplitude $\lambda_0$ nor the order-three automorphism fixes their unequal ratios.  A symmetric Majorana mass matrix may consistently be written as $M_\nu=O_\nu D_\nu O_\nu^T$ with $O_\nu=\exp(P_\nu{\cal B}P_\nu)$, but deriving this map from the GTD condensate remains open.  The consistency statement is therefore directional: the transport model makes the charged-lepton contribution CKM-small, the cyclic lift identifies the CP-safe neutral support, and the observed large angles must still be selected by neutral-sector vacuum dynamics.

\subsection{Internal correlations and future tests}

(i) After the stated transport and amplitude postulates are adopted, $\Vcb=0.0422$ has no fitted continuous CKM parameter in the all-theory frame, but its stability is explicitly hostage to the parent's $\sqrt{m_c/m_t}$ (elasticity $-1.9$; Sec.~\ref{sec:kappa}): the sharp statement is a \emph{correlation}, not a number. Any refinement of the programme's mass matching that moves $\sqrt{m_t/m_c}$ toward its measured value must move $\Vcb$ away from $0.042$ unless the virtual-node reading is simultaneously corrected, and this coupled motion is testable within the programme. (ii) The same correlation logic applies to $\Vus\leftrightarrow\sqrt{m_s/m_d}$ (a $5\%$-level effect): scheme refinements must move both together. (iii) Within the transport family, the relative reality of the $(2,3)$ blocks is required by data (Lemma~\ref{lem:nogo}); a future shift of the $\Vcb$--$\Vub$ pair outside the band covered by the lemma's trade-off would falsify the family outright. (iv) The bridge computation owes two continuous numbers, $(\varepsilon,\omega)$, equivalently $(\theta_{13},\delta_{CP})$, and must also select one of the two balanced-orientation branches.  The raw $\omega$ and the branch-A ratio $0.604$ are not invariant targets; the decisive future comparisons are the moduli correlations (i)--(ii) and the microscopic bridge computation itself.

\section{Conclusions}
\label{sec:conclusions}

The operator audit closes one conceptual gap and exposes another. The occupied-chain compression of the multiplicative symmetric-cube operator is a genuine root-mass operator, and its positive square has exactly the proposed mass-ratio spectrum. The exceptional data therefore define a finite Dirac \emph{modulus} without requiring an associative representation of the full Albert algebra. They do not define a unique finite Dirac operator: $Y_f=y_{f3}U_f\mathcal Q_fW_f^\dagger$ still contains the postulated left transport and an uncomputed right frame. This is not a semantic qualification. The natural Hermitian completion is dense, not nine-link, and its selected loop phase moves from $90^\circ$ to $100.03^\circ$ when the fitted long edge is added.

The present Albert calculation nevertheless constructs a nontrivial part of the bridge after an $F_1$--$F_2$ physical-family embedding is adopted.  The inner derivation ${\cal D}_{12}(-g_\chi)$ changes that chosen Peirce slot, reproduces the conjugate up/anti-down amplitudes and promotes $\varphi_{12}=-2\chi$ to an exceptional operator statement.  The split-real refinement adds a distinct result: the four-component complement in $J_2(\mathbb O_s)=J_2(\mathbb H_s)\oplus\mathbb H_s\ell$ enters the Albert cubic as a $\operatorname{Spin}(3,3)$-scalar pairing of opposite Weyl modules.  It is therefore a candidate universal chiral-soldering channel, not a family-mixing field.  Under the explicit factorization assumption $Y_f=B\otimes\Xi_f$, the common $B$ cancels from the family eigenvectors and CKM remains $U_{\Xi_u}^\dagger U_{\Xi_d}$.  This both sharpens the possible origin of the chiral bridge and proves why it does not close the CKM calculation.

For the adopted cyclic Majorana placement and real-linear projection, the Peirce completion has exactly the CP-safe support $(e_4,e_3,e_6)$, with the quark-active $e_1$ quadrature absent.  What remains undetermined is now narrower but still essential: the vacuum selection of the physical Peirce-family embedding, the family kernels $\Xi_{u,d}$, the coloured $23$ and $31$ contractions, the three unequal edge strengths, the right-sector soldering, and the map from the projected antisymmetric generator to the symmetric Majorana mass operator.  Neither the corner scalar $\lambda_0$ nor the universal $B$ channel selects either balanced Layer-2 branch.

The finite-operator formulation also makes the unity with the scaffolding, BF, emergence and mass-ratio constructions explicit. The Preprints.org scaffolding paper uses the $E_8\supset SU(3)_{\rm geom}\times E_6$ branching to separate a proposed split-bioctonionic $(3,3)$ geometric scaffold from the $E_6$ matter fields and relates its four-dimensional fibres to the Albert-algebra mass geometry \cite{Scaffolding}. The present paper sharpens that proposal by identifying the quadratic base with $J_2(\mathbb H_s)$ and the cubic particle invariant with $J_3(\mathbb O_{\mathbb C})$. Their exact corner relation $N_3(\operatorname{diag}(A,\lambda))=\lambda N_2(A)$ shows that the refinement is more than a dimensional analogy, although the required Peirce projector and real-form selection are not yet derived. Before breaking there are two isomorphic copies $J_L\oplus J_R$ of the complex Albert algebra and no preferred charge--mass identification. After the distinct vacuum choices of the parent construction, $J_L$ carries the flavour/charge frame and $J_R$ the square-root-mass frame. The proposed $SO(3,3)$ BF dynamics reduces the $(3,3)$ geometry to two overlapping leaves, with one carrying weak-force geometry \cite{WesleySinghIsidroBF}. Pairing $(\Sigma_L,J_L)$ with the weak/flavour frame and $(\Sigma_R,J_R)$ with the square-root-mass frame therefore gives a coherent geometric interpretation of the finite bridge. It remains a postulate until the $\omega$-dependent real-form map, physical Clifford-module coupling and overlap/gluing kernel are derived.

Within that interpretation, the observed masses are the singular values of the quark blocks of $D_F$, while CKM mixing is the relative orientation of their left spectral projectors, not simply the left--right overlap itself. The same $D_F$ therefore carries both observables, but at different logical levels: the exceptional Jordan spectrum and occupied chains determine its radial data $\mathcal Q_f$; the transport postulates partially specify its angular data $U_f$; and the bifermionic/Higgs bridge must ultimately derive the full chiral kernel, including $W_f$. The colour connection is a commuting octonionic gauge fluctuation on a different tensor factor, and the electroweak vacuum scale enters through the fluctuated finite operator only once. In this form the present paper supplies a sector-dependent replacement for the emergence paper's schematic $Y_f=y_fY_J$ while remaining consistent with its split-biquaternionic kinetic operator and spectral-action traces.

Appendix~\ref{app:vacuum} turns the remaining underdetermination into a conditional vacuum theorem. The displayed spectral traces contain no mixed up--down word and therefore cannot select the CKM orientation. Adding the most general renormalizable orientation coupling for two minimal bifundamentals gives either commuting left mass operators or a flat relative frame; invariance under independent right-family rotations leaves $W_u,W_d$ unobservable to the potential at every polynomial order. Appendix~\ref{app:edgepotential} gives the complementary edge result: the minimal radial-quartic cyclic truncation enforces equal edge magnitudes or leaves their direction flat, whereas the mixed cubic $d(\Phi,H_q,H_q)$ can stably align the bridge with $P_WH_q^\#$ inside the chosen subspace $W$. General $C_3$-invariant quartics need not obey the equal-magnitude result, and stability orthogonal to $W$ is not established. The present sources calculate neither $W$ nor the projected adjoint, so this is a conditional completion with no reduction of the PMNS parameter count. Together the results specify what the full exceptional/GTD derivation must add: a microscopic right-family lock, a charged order parameter whose adjoint selects the edge ratios, and vacuum projectors that make the Peirce basis physical.

This also locates what must lie beyond the standard model if the numerical flavour pattern is to be explained rather than fitted. The standard model can accommodate every observed mass and mixing parameter through freely chosen Yukawa matrices, and an algebraic reformulation can reveal their common operator structure, but neither step supplies a law selecting the observed point in Yukawa-parameter space. Some additional theoretical structure is therefore required for a derivation. It need not introduce new low-energy particles: it could be an algebraic selection rule, a boundary condition, a new symmetry or a microscopic bridge dynamics. The present exceptional-Jordan chains and two-leaf bridge are proposed as that extra structure. Their value will depend on whether the currently postulated projectors and transport rules can themselves be derived and whether the resulting common-scale predictions survive renormalization-group and threshold tests.

Of the three open items of the companion mixing paper, the accounting then reads: the $(2,3)$ normalization follows, with no fitted continuous CKM parameter after the structural choices, from the virtual-node composition under a stated amplitude postulate ($\kappa_{23}^{\rm eff}=0.633$, $\Vcb$ $+0.8\%$ in the all-theory frame, with alternative readings and the $\sqrt{m_c/m_t}$ fragility disclosed); the Cabibbo-block phase is placed at balanced quadrature rather than fitted (at the cost of an owned $+5.3\%$ in $\Vus$); its Peirce-changing generator is constructed for an adopted family embedding; and the long edge is reduced to one fitted complex number on either of two discrete balanced-orientation branches. Thus two of four continuous CKM degrees of freedom are conditional no-fit estimates at $+5.3\%$ and $+0.8\%$, while two are fitted and localized, but the orientation is not dynamically selected. A first-principles bridge computation must now deliver $(\varepsilon,\omega)$ and its orientation, adjudicate the virtual-node amplitude rule, derive the physical Peirce embedding, and fix the chiral projections and right frames. The new lift reveals a plausible connection to Ref.~\cite{AFHM}: a Yukawa block linear in six diagonal mass links and three directed Albert links would form a connected nine-link graph with one cycle, and a sole balanced Cabibbo phase would give a $\pi/2$ holonomy. But this remains a candidate skeleton until the finite Dirac operator is shown to depend linearly on those projected links; the exponential transport and the Hermitian completion are dense. Only after that selection can the local phase be identified with a nine-link Yukawa-loop phase or used in a determinant-level strong-$CP$ claim.

\appendix

\section{Vacuum-selection obstruction for the minimal chiral bridge}
\label{app:vacuum}

The finite-Dirac form \eqref{eq:Ysvd} identifies the data that a microscopic bridge must determine, but it does not supply a vacuum equation for them. This appendix asks a narrower question: if the up- and down-sector chiral blocks are promoted to the minimal dynamical family bifundamentals, can the lowest-order potential select the observed left misalignment, the right frames, and a nine-link basis? The answer is negative under the field content and symmetry stated below. The result is conditional: an exceptional locking tensor or additional family order parameters would evade its assumptions and are precisely the missing structures to be derived.

Let
\begin{equation}
 \Omega_u\longmapsto G_Q\Omega_uG_u^\dagger,
 \qquad
 \Omega_d\longmapsto G_Q\Omega_dG_d^\dagger
 \label{eq:Omegatransform}
\end{equation}
under the kinetic-term family group
$G_F=U(3)_Q\times U(3)_u\times U(3)_d$, and set
\begin{equation}
 X_u=\Omega_u\Omega_u^\dagger,
 \qquad X_d=\Omega_d\Omega_d^\dagger.
 \label{eq:Xudvac}
\end{equation}
The two $X_f$ act on the common left-handed quark-doublet family space. In a singular-value decomposition $\Omega_f=U_fD_fW_f^\dagger$, they are $X_f=U_fD_f^2U_f^\dagger$ and hence know the left frame but not $W_f$.

The heat-kernel traces displayed in the emergence construction, and their sector-resolved form \eqref{eq:abrefine}, are sums of terms such as $\operatorname{Tr}X_f$ and $\operatorname{Tr}X_f^2$. They contain no mixed word in $X_u$ and $X_d$. Consequently they are invariant under an independent change of the relative left orientation,
\begin{equation}
 X_d\longmapsto V X_dV^\dagger,
 \qquad
 \frac{\partial V_{\rm spec}}{\partial V}=0.
 \label{eq:sourceflat}
\end{equation}
Thus the source-level spectral potential leaves $V_{\rm CKM}=U_u^\dagger U_d$ flat. If the schematic source ansatz $Y_f=y_fY_J$ is instead read as one common $Y_J$, the two left mass operators are proportional and give trivial mixing; if it abbreviates distinct $Y_J^u,Y_J^d$, their relative frame is an unspecified input.

\begin{theorem}[Minimal renormalizable alignment]
\label{thm:vacalign}
Assume that $\Omega_u$ and $\Omega_d$ are canonically dimension-one scalars, are the only family-carrying bridge fields, and have nondegenerate nonzero singular values. Up to functions of the separate singular values, the unique renormalizable $G_F$-invariant term sensitive to their relative left orientation is
\begin{equation}
 V_{\rm orient}=\kappa_{ud}\operatorname{Tr}(X_uX_d).
 \label{eq:Vorientation}
\end{equation}
If $\kappa_{ud}\ne0$, every stationary orientation satisfies
\begin{equation}
 [X_u,X_d]=0,
 \label{eq:vaccommute}
\end{equation}
so $V_{\rm CKM}$ is a permutation matrix up to diagonal phases. If $\kappa_{ud}=0$, the relative orientation is a continuous flat direction. Hence this potential has no isolated realistic CKM vacuum.
\end{theorem}

\begin{proof}
Besides products of the separate traces, the only quartic contraction joining the two sectors without identifying the distinct $u_R$ and $d_R$ indices is $\operatorname{Tr}(X_uX_d)$. Terms $\det\Omega_f+\mathrm{h.c.}$ may be admitted after reducing the corresponding $U(3)$ factors to $SU(3)$, but they fix only sectoral determinant phases and singular values. For an infinitesimal relative rotation $X_d\mapsto e^{tK}X_de^{-tK}$, $K^\dagger=-K$,
\begin{equation}
 \left.\frac{dV_{\rm orient}}{dt}\right|_{t=0}
 =\kappa_{ud}\operatorname{Tr}\!\left([X_d,X_u]K\right).
 \label{eq:vacvariation}
\end{equation}
Stationarity for every anti-Hermitian $K$ gives Eq.~\eqref{eq:vaccommute}. Two commuting Hermitian matrices with nondegenerate spectra have a common eigenbasis, up to phases and permutations. If the electroweak completion forbids the mixed contraction in Eq.~\eqref{eq:Vorientation}, it realizes the $\kappa_{ud}=0$ flat case rather than evading the dichotomy. Equivalently,
$\operatorname{Tr}(X_uX_d)=\sum_{ij}a_i b_j|V_{ij}|^2$ is linear in the doubly stochastic matrix $(|V_{ij}|^2)$, and its global extrema occur at permutations.
\end{proof}

There is a separate all-orders obstruction. Under Eq.~\eqref{eq:Omegatransform}, $W_u$ and $W_d$ move independently along the $U(3)_u$ and $U(3)_d$ symmetry orbits. Every invariant polynomial built solely from $\Omega_u,\Omega_d$ and their adjoints must contract a right index back into the same right-sector space. It can therefore depend on trace words in $X_u,X_d$, their singular values and, for $SU(3)$ rather than $U(3)$ factors, sectoral determinant phases, but not on the right eigenvectors. Hence
\begin{equation}
 \frac{\partial V_{\rm inv}}{\partial W_u}
 =\frac{\partial V_{\rm inv}}{\partial W_d}=0
 \qquad\text{at every polynomial order.}
 \label{eq:rightflat}
\end{equation}
This statement would fail if the microscopic theory supplied a tensor that contracts $u_R$ and $d_R$ family indices or locked either right space to an independently selected Jordan frame. Such a tensor is additional bridge structure, not an invariant made from the two minimal bifundamentals alone.

The physical consequence is sharper than the observation that a numerical fit is incomplete. Matrix support changes under
$\Omega_u\mapsto G_Q\Omega_uG_u^\dagger$ and
$\Omega_d\mapsto G_Q\Omega_dG_d^\dagger$; therefore exact texture zeros, a nine-link support graph and its loop phase cannot be ultraviolet statements until the vacuum selects the relevant left and right rank-one projectors. Higher trace words in $X_u,X_d$ may produce noncommuting left frames, but they still cannot determine $W_u,W_d$ under the symmetry above. A derivation of nine-link sparsity consequently requires at least
\begin{equation}
 \{P_i^Q\},\qquad \{P_j^u\},\qquad \{P_k^d\},
 \qquad
 P_i^Q\Omega_fP_j^f=0\ \text{or}\ \ne0
 \label{eq:vacprojectors}
\end{equation}
as vacuum-selected data, together with a right-sector locking or selection rule. Quantization of the surviving loop phase at $k\pi/8$ would further require a derived discrete remnant or holonomy condition; it does not follow from the alignment theorem. Thus the minimal field content explains why the finite-Dirac mass modulus can be derived while the preferred chiral bridge remains open.

\section{Peirce-edge potential: a truncation test and a conditional completion}
\label{app:edgepotential}

The exact generators of Sec.~\ref{sec:albertbridge} reduce the remaining neutral-sector question to three real edge amplitudes.  Write
\begin{equation}
 \boldsymbol\rho=(\rho_{31},\rho_{12},\rho_{23}),
 \qquad r^2=\rho_{31}^2+\rho_{12}^2+\rho_{23}^2.
 \label{eq:rhovector}
\end{equation}
The minimal radial-quartic cyclic truncation considered for an isolated bridge is
\begin{equation}
 V_0(\boldsymbol\rho)
 =-\frac{m^2}{2}r^2+\frac{\lambda_\Phi}{4}r^4
 +\kappa_\Phi\rho_{31}\rho_{12}\rho_{23}.
 \label{eq:Vedgeisolated}
\end{equation}
On the coassociative slice used in the mass-ratio construction the reduced cubic can vanish; retaining it gives the stronger test.  At a full-rank stationary point set $a=-m^2+\lambda_\Phi r^2$.  The three equations are $a\rho_i+\kappa_\Phi\rho_j\rho_k=0$.  Multiplying two equations by their corresponding $\rho_i$ and subtracting gives $a(\rho_i^2-\rho_j^2)=0$.  If every component and $\kappa_\Phi$ are nonzero, $a\ne0$, and therefore
\begin{equation}
 \boxed{|\rho_{31}|=|\rho_{12}|=|\rho_{23}|.}
 \label{eq:edgeequal}
\end{equation}
If $\kappa_\Phi=0$, the radial minimum leaves the direction of $\boldsymbol\rho$ flat.  Thus the specific truncation \eqref{eq:Vedgeisolated} gives equal edge magnitudes or no directional selection; it cannot account for the observed hierarchy of lepton rotations.

This conclusion is deliberately narrower than a no-go theorem for every cyclic quadratic--cubic--quartic potential.  Mere $C_3$ symmetry permits additional invariants, including $\sum_i\rho_i^3$ and oriented quartics such as
\begin{equation}
 \rho_{31}^{3}\rho_{12}+\rho_{12}^{3}\rho_{23}
 +\rho_{23}^{3}\rho_{31},
 \label{eq:additionalcyclic}
\end{equation}
together with the oppositely oriented companion, unless further symmetries or the microscopic field representation forbid them.  Such terms alter the stationary equations and may support unequal full-rank extrema.  Equation~\eqref{eq:edgeequal} is therefore a result for the minimal radial-quartic truncation, not for general $C_3$-invariant vacuum dynamics.

A charged-sector vacuum can break this cyclic symmetry.  Let $\Phi,H_q\in\mathbf{27}$ and let $d$ denote the polarized Albert cubic.  Its derivative defines the Jordan adjoint through
\begin{equation}
 d(\Phi,H_q,H_q)=\langle\bar\Phi,H_q^\#\rangle
 \label{eq:Jordanadjointsource}
\end{equation}
up to normalization.  Provided the remaining gauge and discrete charges allow it, the minimal exceptional mixed term $-\eta\operatorname{Re}d(\Phi,H_q,H_q)$ is a linear source for the bridge after $H_q$ condenses.  On the CP-safe three-edge subspace $W$, define
\begin{equation}
 \mathbf J=\eta P_WH_q^\#,
 \qquad
 V_{\rm cap}(\boldsymbol\rho)
 =-\frac{m^2}{2}r^2+\frac{\lambda_\Phi}{4}r^4
 -\mathbf J\cdot\boldsymbol\rho,
 \qquad \lambda_\Phi>0.
 \label{eq:Vedgecoupled}
\end{equation}
Within $W$, every global minimum has $\boldsymbol\rho=t\widehat{\mathbf J}$, where $t>0$ obeys
\begin{equation}
 \lambda_\Phi t^3-m^2t-\|\mathbf J\|=0.
 \label{eq:edgeradial}
\end{equation}
The two transverse Hessian eigenvalues in $W$ are $\|\mathbf J\|/t>0$ and the radial eigenvalue is $2\lambda_\Phi t^2+\|\mathbf J\|/t>0$.  The source-selected solution is therefore an isolated stable minimum within $W$ and
\begin{equation}
 \boxed{\boldsymbol\rho\parallel P_WH_q^\#.}
 \label{eq:edgealignment}
\end{equation}

This does not establish stability in the full $\mathbf{27}$.  The directions orthogonal to $W$ require independently positive masses, constraints or a microscopic selection rule, and the projector $P_W$ is itself part of the still-uncomputed vacuum data.

This is an existence result, not yet a prediction.  The previous common-vacuum PMNS ansatz would require, in the ordering of Eq.~\eqref{eq:rhovector},
\begin{equation}
 \boldsymbol\rho_\star
 =\frac{1}{\cos\chi}
 \left(-\frac{\pi}{4}+\theta_C^J,\,
       \frac{\pi}{4},\,
       \frac{\theta_C^J}{2}-2\varepsilon_q\right),
 \label{eq:rhotargetpaper}
\end{equation}
where $\theta_C^J$ is the Layer-1 Cabibbo angle and $\varepsilon_q$ the fitted quark long-edge angle.  The missing microscopic theorem is precisely
\begin{equation}
 P_WH_q^\#\parallel\boldsymbol\rho_\star.
 \label{eq:missingadjoint}
\end{equation}
No charged order parameter specified in the present sources has yet been shown to obey Eq.~\eqref{eq:missingadjoint}.  Without it, the three components of $\mathbf J$ provide two directional ratios and one magnitude.  In the PMNS parametrization of Ref.~\cite{SinghNu}, the numerical Jacobian from $(b_4,b_3,b_6)$ in Eq.~\eqref{eq:MajoranaSupport} to $(\theta_{12},\theta_{23},\theta_{13})$ has determinant $0.935$ at the proposed point, so the three source components are locally equivalent to fitting the three angles.  The coupled potential shows how unequal edges can arise, but does not reduce the parameter count until $H_q^\#$ is independently calculated.  The corner scalar $\lambda_0$ can alter the radial normalization in Eq.~\eqref{eq:edgeradial}; it cannot determine the two edge ratios.

\section{Proofs and closed forms}
\label{app:proofs}

\emph{Postulate~\ref{post:comp}: the projected block and the alternative readings.} Work in the three-dimensional space (gen-2, virtual, gen-3). The composed transport is $T=R_{v3}(\theta_3)R_{2v}(\theta_2)$; its projection onto the physical states is the triangular block
$\big(\begin{smallmatrix}\cos\theta_2 & 0\\ \sin\theta_2\sin\theta_3 & \cos\theta_3\end{smallmatrix}\big)$,
with $T_{32}=\sin\theta_2\sin\theta_3$ the two-step amplitude and $T_{23}=0$. The block is not unitary (the discarded virtual admixture is $\sin^2\theta_2\cos^2\theta_3=26\%$ of the gen-2 norm), and different re-unitarizations assign different angles: polar decomposition $0.0675$, row normalization $0.1257$, column normalization $0.1435$, amplitude reading $0.1232$; the corresponding $\Vcb$ values are $0.014$, $0.045$, $0.063$, $0.042$. The amplitude reading is adopted as Postulate~\ref{post:comp} for the reasons stated there. With $\tan\theta_2=1/(1+\delta)$ and $\tan\theta_3=(1-\delta)/(1+\delta)$: $\sin\theta_2=0.527065$, $\sin\theta_3=0.233746$, $\sin\theta^d_{23}=0.123200$.

\emph{Lemma~\ref{lem:fritzsch}.} With $U_f=R_{23}R_{12}$ and phases only on the $(1,2)$ blocks, the exact matrix elements are $V_{ub}=-s^{u}_{12}e^{-i\varphi}D$ and $V_{cb}=c^{u}_{12}D$ with the common factor $D=s^{d}_{23}c^{u}_{23}-c^{d}_{23}s^{u}_{23}$, whence $|V_{ub}|/|V_{cb}|=\tan\theta^u_{12}=\sqrt{m_u/m_c}$ exactly (verified to machine precision); numerically $0.042449\times0.042202=0.001791$.

\emph{Wolfenstein and unitarity-triangle extraction.} Standard definitions \cite{Wolfenstein,Jarlskog} are used: $\lambda=s_{12}$, $A=s_{23}/\lambda^2$ (tabulated via the proxy $|V_{cb}|/|V_{us}|^2$, equal to it up to $\mathcal O(s_{13}^2)$ corrections), $\bar\rho+i\bar\eta=-(V_{ud}V^*_{ub})/(V_{cd}V^*_{cb})$, $\beta=\arg[-(V_{cd}V^*_{cb})/(V_{td}V^*_{tb})]$, $\gamma=\arg[-(V_{ud}V^*_{ub})/(V_{cd}V^*_{cb})]$, and $J=\mathrm{Im}(V_{us}V_{cb}V^*_{ub}V^*_{cs})$.

\section{Numerical documentation}
\label{app:numerics}

All CKM-transport numbers are produced by the ancillary script \texttt{ckm\_analysis.py} (included). In addition to the CKM calculations described below, the script enumerates the ten normalized occupation labels of $\Sym^3(\mathbf3)$, constructs the occupied-node operators $\mathcal R_{u,d}$ and $\mathcal Q_{u,d}$, verifies the six eigenvalues of each sectoral finite Dirac block, and reconstructs the Hermitian representative used for the eighteen-entry and loop-phase statements of Sec.~\ref{sec:finite}. It evaluates the rung tangents \eqref{eq:rungs} from $\delta=\sqrt{3/8}$, the composition \eqref{eq:comp} together with the alternative re-unitarization readings and the input-frame table of Sec.~\ref{sec:kappa}, the Layer-1 assembly and Table~\ref{tab:layer1}, the no-go scan of Lemma~\ref{lem:nogo}, the failed-texture scan of Appendix~\ref{app:failed} (code included), and both Layer-2 orientation branches. For each orientation, the two fitted parameters are obtained by solving exactly on the two fitted observables: $\varepsilon$ by bisection on $\Vub$ at each $\omega$, and $\omega$ by bisection on $\delta_{CP}$ over the relevant monotone interval.  Orientation A gives $(\varepsilon,\omega)=(0.00207936,288.1915^\circ)$ and orientation B gives $(0.00535850,202.1844^\circ)$, with both targets reproduced to machine precision and positive $J$. The $CP$ phase is extracted quadrant-safely by construction, $\delta_{CP}=\mathrm{atan2}(\sin\delta,\cos\delta)$ with $\sin\delta$ from the Jarlskog invariant and $\cos\delta$ from $|V_{td}|^2$ via $\cos\delta=(s_{12}^2s_{23}^2+c_{12}^2c_{23}^2s_{13}^2-|V_{td}|^2)/(2s_{12}s_{23}c_{12}c_{23}s_{13})$. Experimental targets are the PDG 2026 CKM-review global-fit Wolfenstein parameters $\lambda=0.22517$, $A=0.826$, $\bar\rho=0.1576$, $\bar\eta=0.3556$ \cite{PDG}, with \emph{all} remaining targets derived from these four through the exact standard parametrization, so that the target set is internally coherent: $\Vus=0.22517$, $\Vcb=0.041879$, $\Vub=0.0037623$, $\Vtd=0.008625$, $\Vts=0.041154$, $J=3.158\times10^{-5}$, $\delta_{CP}=66.13^\circ$, $\beta=22.89^\circ$, $\gamma=66.10^\circ$. The quoted deviations use these central values; no theory uncertainty is assigned to the model numbers, global-fit uncertainties are at or below the percent level for the well-measured entries \cite{PDG,CKMfitter,UTfit}, and the percentage deviations are therefore theory-dominated.  The Layer-1 mass-ratio comparison uses the common parent baseline $\mu_*=M_Z$; it does not mix mass inputs quoted at different scales. The separate \texttt{nine\_link\_jordan\_scan.py} performs the exact support enumeration and reduced-protocol test of Sec.~\ref{sec:afhm}, using the $M_Z$ control values of Ref.~\cite{AFHM} and the eigenvalue ratios of $\mathcal Q_{u,d}$. The archived raw scans and path-independent summarizer reconstruct the quoted central-value statistics; root discovery remains dependent on the documented multistart protocol and software environment.

\section{The failed Hamiltonian-texture extension}
\label{app:failed}

For completeness: reading (i) of Sec.~\ref{sec:transport} was implemented as the family of symmetric tridiagonal matrices on the down chain with zero diagonal except the final node, hops solved in closed form from the characteristic polynomial (matching the three physical eigenvalues and one auxiliary eigenvalue assigned to the virtual node), all $2^4$ sign patterns and auxiliary magnitudes $\{ac^2\ \text{monomial},\ \tfrac12 ac^2,\ 2ac^2,\ (1+\delta)^2,\ m_b\}$ scanned. The corrected characteristic-polynomial inversion (an earlier version of this scan carried a sign error in the $E_3/E_1$ terms, corrected here) yields ten admissible sign assignments; across them $\Vus$ ranges over $0.05$--$0.66$ and the non-unitarity (virtual leakage) over $20$--$69\%$. The assignment closest to the Cabibbo value, $\Vus=0.27$, carries $31\%$ non-unitarity; none falls below $20\%$. The exclusion is therefore decisive for every tested case, but we state its scope plainly: it covers the enumerated sign patterns and the five auxiliary-eigenvalue choices, not a continuous scan of the texture family (the included script reproduces the scan). The same conclusion holds with the phases of Postulate~\ref{post:phase} in place. The transport reading, by contrast, preserves the two-state relations exactly and is unitary by construction. This comparison is recorded so that the grounds for the modelling choice are explicit; as stated in Sec.~\ref{sec:transport}, the exclusion covers this texture family, not reading (i) in full generality.

\section*{Acknowledgements}

\noindent{\bf Use of generative AI}: This document has been prepared with research assistance from Anthropic's Claude Fable 5. Independent technical and adversarial reviews by OpenAI's GPT-5.6 Sol identified the finite-Dirac operator audit, the gauge-completed operator dictionary and quadratic--cubic Jordan clarification of Sec.~\ref{sec:finite}, the exact Peirce-changing Albert lift of Sec.~\ref{sec:albertbridge}, the split-real Peirce decomposition and conditional chiral-factorization result of Sec.~\ref{sec:universalsoldering}, the vacuum-selection analyses of Appendices~\ref{app:vacuum} and \ref{app:edgepotential}, the sign error of Appendix~\ref{app:failed}, and the coherent-target refit of Appendix~\ref{app:numerics}. A further adversarial review by Claude Fable 5 identified the discrete Layer-2 branch ambiguity and related scope and convention corrections incorporated in this version. The author takes full intellectual responsibility for the contents of this article.

\end{document}